\pdfoutput=1
\documentclass[11pts]{article}

\usepackage{enumerate}
\usepackage{hyperref}
\usepackage[utf8]{inputenc}
\usepackage[T1]{fontenc}
\usepackage[english]{babel}
\usepackage{url}
\usepackage[squaren, Gray, cdot]{SIunits}
\usepackage{array}
\usepackage{caption}
\usepackage{layout}
\usepackage{csquotes}
\usepackage[top=2.8cm,bottom=2.8cm,left=2.8cm,right=2.8cm]{geometry}

\usepackage{amsthm}
\usepackage{amsmath}
\usepackage{amssymb}
\usepackage{thm-restate}

\newtheorem{theorem}{Theorem}
\newtheorem{definition}[theorem]{Definition}
\newtheorem{lemma}[theorem]{Lemma}
\newtheorem{proposition}[theorem]{Proposition}
\newtheorem{claim}[theorem]{Claim}
\newtheorem{remark}[theorem]{Remark}

\usepackage{listings}

\usepackage{graphicx}
\usepackage{subfigure}

\usepackage{bbm}
\usepackage{stmaryrd}
\usepackage{tikz}

\usepackage{algorithm}
\usepackage[noend]{algpseudocode}

\newcommand{\rnk}{\mathrm{rank}}
\newcommand{\spn}{\mathrm{span}}

\allowdisplaybreaks

\title{Robust Sparsification for Matroid Intersection with Applications}

\date{}

\author{Chien-Chung Huang\\
\textit{CNRS, DI ENS, École normale supérieure, Université PSL, Paris, France}
\and François Sellier\\
\textit{DI ENS, École normale supérieure, Université PSL, Paris, France}\\
\textit{Université Paris Cité, CNRS, IRIF, Paris, France}\\
\textit{Mines Paris, Université PSL, Paris, France}}

\begin{document}

\maketitle
\begin{abstract}
    
    Matroid intersection is a classical optimization problem where, given two matroids over the same ground set, the goal is to find the largest common independent set. In this paper, we show that there exists a certain ``sparsifer'': a subset of elements, of size $O(|S^{opt}| \cdot 1/\varepsilon)$, where $S^{opt}$ denotes the optimal solution, that is guaranteed to contain a $3/2 + \varepsilon$ approximation, while guaranteeing certain robustness properties. We call such a small subset a \emph{Density Constrained Subset} (DCS), which is inspired by the \emph{Edge-Degree Constrained Subgraph} (EDCS) [Bernstein and Stein, 2015], originally designed for the maximum cardinality matching problem in a graph. Our proof is constructive and hinges on a greedy decomposition of matroids, which we call the \emph{density-based decomposition}. We show that this sparsifier has certain robustness properties that can be used in one-way communication and random-order streaming models. 
    
    Specifically, we use the DCS to design a one-way communication protocol for matroid intersection and obtain a $3/2 + \varepsilon$ approximation, using a message of size $O(|S^{opt}| \cdot 1/\varepsilon)$. This matches the best achievable ratio for the one-way communication bipartite matching [Goel, Kapralov, and Khanna, 2012]. 
    
    Moreover, the DCS can be used to design a streaming algorithm in the random-order streaming model requiring the space of $O(|S^{opt}|\cdot poly(\log(n), 1/\varepsilon))$, where $n$ is the size of the stream (the ground set of the matroids). Our algorithm guarantees a $3/2 + \varepsilon$ approximation \emph{in expectation} and, when the size of $S^{opt}$ is not too small, \emph{with high probability}. Prior to our work, the best approximation ratio of a streaming algorithm in the random-order streaming model was an expected $2- \delta$ for some small constant $\delta>0$ [Guruganesh and Singla, 2017].
    

\end{abstract}

\section{Introduction}

    The matroid intersection problem is a fundamental problem in combinatorial optimization. In this problem we are given two matroids $\mathcal{M}_1 = (V, \mathcal{I}_1)$ and $\mathcal{M}_2 = (V, \mathcal{I}_2)$, and the goal is to find the largest common independent set in both matroids, \emph{i.e.}, $\mathop{\arg\max}_{S \in \mathcal{I}_1 \cap \mathcal{I}_2}|S|$. This problem was introduced and solved by Edmonds~\cite{edmonds1970submodular, Edmonds1971, Edmonds1979} in the 70s. The importance of matroid intersection stems from the large variety of combinatorial optimization problems it captures; well-known examples in computer science include bipartite matching and packing of spanning trees/arborescences. 
    
    In this paper we introduce a ``sparsifer'' for the matroid intersection problem and use it to design algorithms for two problems closely related to streaming: a one-way communication protocol and a streaming algorithm in the random-order streaming model. 
    
    \paragraph*{Structural Result for a Matroid Intersection Sparsifier}
    
    Our starting point is the maximum matching problem. To deal with massive graphs, a common tool is sparsification, \emph{i.e.}, the extraction of a subgraph having fewer edges but preserving some desired property. Various graph sparsifiers have been introduced to maintain a large matching, \emph{e.g.}, see~\cite{BernsteinS15,Bhattacharya18,RoghaniSW22,Gupta13} and the references therein. 
    The particular sparsifier that has inspired our work is the
    \emph{Edge-Degree Constrained Subgraph} (EDCS) introduced by Bernstein and Stein~\cite{BernsteinS15}. 
    \begin{definition}[from \cite{BernsteinS15}]
        \label{def:intro-edcs}
         Let $G = (V,E)$ be a graph, and $H$ a subgraph of $G$. Given any integer parameters $\beta \geq 2$ and $\beta^- \leq \beta - 1$, we say that a subgraph $H = (V, E_H)$ is a \emph{$(\beta, \beta^-)$-EDCS} of $G$ if $H$ satisfies the following properties (for $v \in V$, $\deg_H(v)$ denotes the degree of $v$ in $H$):
        \begin{enumerate}[(i)]
            \item \makebox[12em][l]{For any edge $(u, v) \in H$,} $\deg_H(u) + \deg_H(v) \leq \beta$;
            \item \makebox[12em][l]{For any edge $(u, v) \in G \backslash H$,} $\deg_H(u) + \deg_H(v) \geq \beta^-$.
        \end{enumerate}
    \end{definition}
    
    The size of an EDCS is easily controlled by the parameter $\beta$ as it is $O(\beta \cdot |M_G|)$, where $M_G$ is the maximum matching. The key property of EDCSes is that, by choosing some $\beta$ and $\beta^-$ in the order of $O(poly(1/\varepsilon))$, an EDCS is guaranteed to contain a $3/2 + \varepsilon$ approximation of the maximum matching~\cite{AssadiB19,BernsteinS16}. As a result, EDCSes have been used to approximate maximum matching in the dynamic, random-order streaming, communication, and sublinear settings with success, for instance see~\cite{AzarmehrB23,BehnezhadK22,BhattacharyaKS23,BehnezhadRR23,bernstein:LIPIcs:2020:12419,BernsteinS15,BernsteinS16,Grandoni2022,Kiss22}. 
    
    As bipartite graph matching is a special case of matroid intersection, the special case when both matroids are partition matroids, one is naturally prompted to ask: is there an analogue of EDCS for general matroid intersection? However, even very slight generalizations of partition matroids, such as laminar matroids (\emph{i.e.}, adding nested cardinality constraints on groups of vertices on each side of the bipartite graph), it is already unclear how to properly define the equivalent of EDCSes. In fact, to the best of our knowledge, in this setting of laminar matroids, nothing is known about getting approximation ratios comparable to those for simple matching in random streams~\cite{bernstein:LIPIcs:2020:12419} or in communication complexity~\cite{AssadiB19}.

    To properly generalize EDCS, the first question would be: what could be the equivalent of a vertex degree in a graph, in the context of a matroid? To answer this question, we make use of the notion of \emph{density} of a subset in a matroid and introduce the \emph{density-based decomposition}.\footnote{This decomposition is closely related to the notion of \emph{principal sequence}~\cite{Fujishige08}; this aspect will be discussed later.} In the following discussion, we assume that readers are familiar with matroids. All the formal definitions can be found in Section~\ref{sec:density-decomposition}.
    \begin{restatable}[]{definition}{densitydefinition}
    \label{def:density}
        Let $\mathcal{M} = (V, \mathcal{I})$ be a matroid. The \emph{density} of a subset $U \subseteq V$ in $\mathcal{M}$ is defined as \[\rho_{\mathcal{M}}(U) = \frac{|U|}{\rnk_{\mathcal{M}}(U)}.\]
        By convention, the density of an empty set is $0$, and the density of a non-empty set of rank $0$ is $+\infty$.
    \end{restatable}
    
    We now explain, at a high level, how densities are used. Let $V' \subseteq V$ be a subset of elements ($V'$ is meant to be our ``sparsifier'') and consider the matroid $\mathcal{M}'$, which is the original matroid $\mathcal{M}$ restricted to $V'$. 
    Then we apply the following greedy procedure: find the densest set $U_1 \subseteq V'$ and then contract $\mathcal{M}'$ by $U_1$; next find the densest set $U_2 \subseteq V' \backslash U_1$ in the contracted matroid $\mathcal{M}' / U_1$ and again contract  
    $\mathcal{M}' / U_1$ by $U_2$, and so on (for more details about this method and the contraction of a matroid, we refer the reader to Section~\ref{sec:density-decomposition}). This greedy procedure induces a \emph{density-based decomposition} of $V' = U_1 \cup \cdots \cup U_k$, where $k$ is the rank of the original matroid $\mathcal{M}$ (note that some of the last $U_i$s could be empty; to give a better intuition about this decomposition an example is provided in Figure~\ref{fig:laminar-matroid-decomposition}). As a result, each element of $V'$ can be assigned a density based on this decomposition, namely, the density of the set $U_i$ where it appears in, computed with respect to the contracted matroid that was used for the construction of that $U_i$. Each element $v \in V \backslash V'$ can also be assigned a density, namely, the density of the elements in the first set $U_i$ such that $v$ is spanned by $U_1 \cup \dots \cup U_i$ in the matroid $\mathcal{M}$.
    
    Therefore using this notion of \emph{density-based decomposition} of $V'$ in the restricted matroid $\mathcal{M}'$ we can define for every element $v \in V$ an \emph{associated density} $\tilde{\rho}_{\mathcal{M}}(v)$ with respect to $V'$ (a formal definition is provided in Definition~\ref{def:tilderho}). This associated density plays the role analogous to the vertex degree in a graph.
    With the associated densities of the elements, we can define a \emph{Density-Constrained Subset} (DCS) for matroid intersection:
    
    \begin{restatable}[]{definition}{dcsdefinition} \label{def:dcs}
        Let $\mathcal{M}_1 = (V, \mathcal{I}_1)$ and $\mathcal{M}_2 = (V, \mathcal{I}_2)$ be two matroids. Let $\beta$, $\beta^-$ be two integers such that $\beta \geq \beta^- + 7$. A subset $V' \subseteq V$ is called a \emph{$(\beta, \beta^-)$-DCS} if it satisfies the following properties:
        \begin{enumerate}[(i)]
            \item \makebox[8em][l]{For any $v \in V'$,} $\tilde{\rho}_{\mathcal{M}_1}(v) + \tilde{\rho}_{\mathcal{M}_2}(v) \leq \beta$;
            \item \makebox[8em][l]{For any $v \in V \backslash V'$,} $\tilde{\rho}_{\mathcal{M}_1}(v) + \tilde{\rho}_{\mathcal{M}_2}(v) \geq \beta^-$.
        \end{enumerate}
    \end{restatable}
    
    By a constructive proof, we show that such $(\beta, \beta^-)$-DCSes always exist (Theorem~\ref{thm:construction}). This proof is based on a local search argument similar to that of~\cite{AssadiB19} but here it requires to understand how the density-based decomposition of $V'$ is affected when an element is added or removed from $V'$ --- hence we need the two important ``modification lemmas'', namely, Lemmas~\ref{lem:add-increase} and~\ref{lem:del-decrease}. We also prove that DCSes are compact, in the sense that their size is up to $\beta$ times the cardinality of the optimal solution (Proposition~\ref{prop:size-bound}). Moreover, DCSes always contain a good approximation of the optimal solution:
    
    \begin{restatable}[]{theorem}{ratiotheorem} \label{thm:dcs-ratio}
		Let $\varepsilon > 0$. For integers $\beta$, $\beta^-$ such that $\beta \geq \beta^- + 7$ and $(\beta^- - 4) \cdot (1 + \varepsilon) \geq \beta$, any $(\beta, \beta^-)$-DCS $V'$ contains a $3/2 + \varepsilon$ approximation of the maximum cardinality common independent set.
	\end{restatable}
	
	Theorem~\ref{thm:dcs-ratio} can be compared to the result for EDCSes in bipartite graphs:
	\begin{theorem}[from~\cite{AssadiB19}] \label{thm:intro-general-edcs}
        Let $\varepsilon > 0$. For integers $\beta$, $\beta^-$ such that $\beta \geq \beta^- + 1$ and $\beta^- \geq \beta \cdot (1 - \varepsilon/4)$, any $(\beta, \beta^-)$-EDCS $H$ of a bipartite graph $G$ contains a $3/2 + \varepsilon$ approximation of the maximum matching.
    \end{theorem}
	
	The proof of Theorem~\ref{thm:dcs-ratio} is the most crucial part of our work. In the following we briefly discuss our methodology and highlight the important ideas in our proof.
	
	Many algorithms for optimization problems are analyzed based on primal-duality of linear programs. Even though the convex hull of common independent sets can be described by a linear program~\cite{Schrijver2003}, we choose not to use its dual program. Instead we use the simpler mini-max theorem of Edmonds~\cite{edmonds1970submodular}.

	\begin{theorem}[Matroid intersection theorem~\cite{edmonds1970submodular}] \label{thm:matroid-intersection}
		Given two matroids $\mathcal{M}_1 = (V, \mathcal{I}_1)$ and $\mathcal{M}_2 = (V, \mathcal{I}_2)$, the maximum size of a set in $\mathcal{I}_1 \cap \mathcal{I}_2$ is
		\[\min_{U \subseteq V} (\rnk_{\mathcal{M}_1}(U) + \rnk_{\mathcal{M}_2}(V \backslash U)).\]
	\end{theorem}
	
	The minimizers $U$ and $V\backslash U$ in the formula will serve as the ``dual'' to bound the size of the optimal solution. In particular, in our proof, we consider the two matroids $\mathcal{M}'_1$ and $\mathcal{M}'_2$ derived from original matroids $\mathcal{M}_1$ and $\mathcal{M}_2$ restricted to $V'$.  
	Edmonds' theorem states that we can find $C_1$ and $C_2$ 
	so that $C_1 \cup C_2 = V'$, $C_1 \cap C_2 = \emptyset$, and  $\rnk_{\mathcal{M}'_1}(C_1) + \rnk_{\mathcal{M}'_2}(C_2)$ is equal to the size of the maximum common independent set in $V'$, denoted as $\mu(V')$. The question then boils down to compare the size of an optimal solution with $\rnk_{\mathcal{M}'_1}(C_1) + \rnk_{\mathcal{M}'_2}(C_2)$. 
	
	To achieve this, we will use a certain greedy procedure to choose a subset $S$ of the optimal solution so that in the contracted matroids $\mathcal{M}_1 / S$ and $\mathcal{M}_2 / S$, all the remaining elements of the optimal solution not in $S$ are spanned in at least one of these contracted matroids (either by $C_1$ in $\mathcal{M}_1 / S$ or by $C_2$ in $\mathcal{M}_2 / S$). 
	As these elements are of size at most $\rnk_{\mathcal{M}_1/S}(C_1) + \rnk_{\mathcal{M}_2/S}(C_2) \leq \rnk_{\mathcal{M}'_1}(C_1) + \rnk_{\mathcal{M}'_2}(C_2)=\mu(V')$ (by Edmonds' theorem), 
	we just need to bound the size of $S$. We will use a strategy to bound the size of $S$ by $(1/2 +\varepsilon) \cdot (\rnk_{\mathcal{M}'_1}(C_1) + \rnk_{\mathcal{M}'_2}(C_2))$. That part of the proof hinges on the construction of well-chosen subsets of $C_1$ and $C_2$ (see Lemmas~\ref{lem:defineR} and~\ref{lem:defineQ}), and on the properties of the density-based decomposition. In fact, in the case of graph matching (two partition matroids), the proof for EDCSes can be done by an edge counting argument~\cite{AssadiB19}, whereas here we need a more sophisticated proof strategy--- how the density decomposition is useful and exploited is fully displayed in the proofs of Lemmas~\ref{lem:defineR} and~\ref{lem:defineQ}. 
	
	\begin{remark}
	    When the two matroids are partition matroids of the same rank, the definition of associated density for an element matches the notion of degree for the endpoint of an edge: hence in that case our DCS definition corresponds to that of EDCS in a bipartite graph.
	\end{remark}

	\paragraph*{Application to One-Way Communication}
	
	We consider the following one-way communication problem~\cite{kushilevitz1997communication}: Alice is given some part $V_A$ of the common ground set $V$, while Bob holds the other part $V_B$. The goal for Alice is to send a single message to Bob so that Bob outputs an approximate maximum common independent set. If Alice sends her whole ground set $V_A$ to Bob, then the latter will be able to recover the exact solution. However in this game, we assume that communication is costly, so we would like to do as best as possible while restricting ourselves to use a message of size $O(\mu(V))$, where $\mu(V)$ denotes the size of the optimal solution. For instance, if Alice sends only a maximum intersection in $V_A$ then Bob is able to complete it to make it a maximal set (a set such as no element can be added to it without creating a circuit in one of the two matroids), and we then obtain a $2$ approximation protocol. The interest in studying one-way communication problems lies in their connection with the single-pass streaming model~\cite{GoelKK12} and other computational models, as they, in a certain way, capture the essence of trade-offs regarding message sizes.
	
    Our problem is a natural generalization of the one-way communication problem for matchings, which has been studied in~\cite{AssadiB19, GoelKK12}: the edges of the graph are splitted by some adversary between Alice and Bob, and Alice has to send a small message to Bob so he can recover some good matching.  In particular, when both matroids are partition matroids, our problem is equivalent to the one-way communication in a bipartite graph. Protocols have been provided for the one-way communication matching problem to get a $3/2$ approximation, see~\cite{AssadiB19, GoelKK12}. Moreover, we know that for bipartite graphs with $k$ vertices on each side, any protocol providing an approximation guarantee better than $3/2$ requires a message of size at least $k^{1 + \Omega(1/\log \log k)}$~\cite{GoelKK12}. Therefore in our general case of matroid intersection one cannot expect to beat the $3/2$ approximation ratio using a message of size $O(\mu(V))$.
    
    Assadi and Bernstein~\cite{AssadiB19} used the EDCS sparsifier to get the optimal $3/2$ approximation ratio.  In Section~\ref{sec:communication} we show that our DCS sparsifier has the same robustness property: if Alice builds some DCS and sends it to Bob, Bob will be able to get an approximate solution with a ratio close to $3/2$. Proving this requires only a slight adaptation of the proof of Theorem~\ref{thm:dcs-ratio}.
	
	\begin{restatable}[]{theorem}{communicationtheorem} \label{thm:communication}
        There exists a one-way communication protocol that, given any $\varepsilon > 0$, computes a $3/2 + \varepsilon$ approximation to maximum matroid intersection using a message of size $O(\mu(V)/\varepsilon)$ from Alice to Bob, where $\mu(V)$ denotes the size of the optimal solution of the matroid intersection problem.
    \end{restatable}
    
    Hence our result closes the gap between matching and matroid intersection, and matches the $3/2$ bound for bipartite matching. It shows that matroid intersection and matching problems have similar one-way communication limitations, despite the more complex structure of matroids.

	\paragraph*{Application to Random-Order Streams}
	
	The \emph{streaming} model of computation~\cite{FKMSZ2005} has been motivated by the recent rise of massive datasets, where we cannot afford to store the entire input in memory. 
    Given that the ground set is made of $|V| = n$ elements, in the streaming model $V$ is presented to the algorithm as a stream of elements $v_1, \dots, v_n$. The algorithm is allowed to make a single pass over that stream and, ideally, uses a memory roughly proportional to the output size (up to a poly-logarithmic factor): therefore the main challenge in this model is that we have to discard many elements through the execution of the algorithm.
    
    We note that, in the most general model where an adversary decides the order of the elements, it has been a long-standing open question whether the maximum matching in bipartite graphs (a very simple case of matroid intersection) can be approximated within a factor better than $2$, \emph{i.e.}, the ratio achievable by the simple greedy algorithm.
    
    Our focus here is on the \emph{random-order} streaming model, where the permutation of the elements of $V$ in the stream 
    is chosen uniformly at random. This is a natural assumption as real-world data has little reason of being ordered in an adversarial way (even though the distribution may not be entirely random either). 
    In fact, as mentioned in~\cite{KonradMM12}, the random-order streaming model might better  explain why certain algorithms perform better in practice than their theoretical bounds under an adversary model.
    It is noteworthy that under the random-order streaming model, for the maximum 
    matching, quite a few recent papers have shown that the approximation factor of 2 can be beaten~\cite{KonradMM12,GamlathSSS2019,Konrad18,FarhadiHMRR20,bernstein:LIPIcs:2020:12419,HuangS2022-edcs}. In addition, in the adversary model, Kapralov~\cite{Kapralov2021} shows that to get an approximation factor  better than $1 + \ln 2 \approx 1.69$, one needs $k^{1+\Omega(1/\log \log k)}$ space, even in bipartite graphs (here $k$ denotes the number of vertices on each side). The paper of Bernstein~\cite{bernstein:LIPIcs:2020:12419} proves that it is possible to beat this adversarial-order lower bound in the random-order model, by achieving a $3/2 + \varepsilon$ approximation while using only $O(k \cdot poly(\log(k), 1/\varepsilon))$ space, thus demonstrating a separation between the adversary model and the random-order model.
    
    For our main topic, matroid intersection, a simple greedy algorithm gives again an approximation ratio of $2$. Guruganesh and Singla~\cite{GuruganeshS17} have shown that it is possible to obtain the factor of $2 - \delta$ in expectation, for some small $\delta > 0$.\footnote{It should be emphasized that Guruganesh and Singla consider the more stringent ``online'' model.} We show that this factor can be significantly improved. In fact, in Section~\ref{sec:streaming}, we use our DCS construction in the context of random-order streams to design an algorithm. The framework developed in Section~\ref{sec:streaming} is a slight modification of that  of~\cite{bernstein:LIPIcs:2020:12419,HuangS2022-edcs}. 
	
	\begin{restatable}[]{theorem}{streamingtheorem}
	    \label{thm:intro-streaming-approx}
	    Let $1/4 >\varepsilon > 0$. One can extract from a randomly-ordered stream of elements a common independent subset in two matroids with an approximation ratio of $3/2 + \varepsilon$ in expectation, using $O(\mu(V)\cdot \log(n) \cdot \log(k) \cdot (1/\varepsilon)^3)$ memory, where $\mu(V)$ denotes the size of the optimal solution, and $k$ is the smaller rank of the two given matroids. Moreover the approximation ratio is worse than $3/2 + \varepsilon$ only with probability at most $\exp(- 1/32 \cdot \varepsilon^2 \cdot \mu(V)) + n^{-3}$.
	\end{restatable}
	
	Thus, not only do we improve upon the factor $2 - \delta$~\cite{GuruganeshS17}, but 
	also we demonstrate that it is possible to beat the adversarial-order lower bound of $1 + \ln 2 \approx 1.69$ of~\cite{Kapralov2021} for the matroid intersection problem as well in the random order model (assuming that $n$ is polynomial in $k$). 
	
	\begin{remark}
	    When the size the the optimal solution $\mu(V)$ is $\Omega(\log(n)/\varepsilon^2)$, we obtain a good approximation ratio with high probability, as the probability of failure will be $n^{-O(1)}$ (and $n$ is assumed to be very big as we are in the streaming setting).
	    Unlike in~\cite{bernstein:LIPIcs:2020:12419,HuangS2022-edcs}, we cannot guarantee with high probability a good approximation ratio when the solution is small: in fact, when a matching is relatively small we can prove that the graph has a limited number of edges (so we can afford to store all of them), but for the matroid intersection problem, a small maximum intersection of two matroids does not imply that the ground set is small as well.
	\end{remark}
	
    \paragraph*{Density-Based Decomposition and Principal Partitions}
    The notion of densest subsets and density-based decompositions is closely related to the theory of \emph{principal partitions}. The latter indeed comes from a long line of research in various domains, ranging from graphs, matrices, matroids, to submodular systems. We refer to readers to a survey of Fujishige~\cite{Fujishige08}. Below we give a quick outline. 

    Let $V'$ be the ground set of a matroid $\mathcal{M}$. By the theory of principal partitions, there exist a sequence of nested sets, called \emph{principal sequence}, $F_1 \subset F_2 \subset \dots \subset F_k=V'$, and a sequence of critical values $\lambda_1 > \lambda_2 > \cdots > \lambda_k$, so that the matroid obtained by contracting $F_{i-1}$ and restricted to $F_i$, is ``uniformly dense'' (\emph{i.e}, no set has a larger density than the ground set itself), with density $\lambda_i$. In our context, recall that $V'$ is decomposed into $U_1, U_2, \dots, U_k$ by a greedy procedure. Then it can be seen that $F_1 = U_1, F_2 = U_1 \cup U_2, \dots, F_k = U_1 \cup \dots \cup U_k$. In this sense, our density-based decomposition can be regarded as a rewriting of the principal sequence, and some basic results stated in Section~\ref{sec:density-decomposition} are already known in the context of principal partitions. However, we adopt this term and this way of decomposing the elements to better emphasize the ``greedy'' nature of our approach and to facilitate our presentation.

    The most important consequence of the theory of principal partitions for us is that the densest sets $U_1, \dots, U_k$ in our greedy procedure can be computed in polynomial time by using submodular function minimization~\cite{Fujishige08}. We briefly explain how it can be done. For any density $\rho$, we can find in polynomial time the largest set $U_{\rho}$ minimizing the submodular function $f_{\rho}(U) = \rho \cdot \rnk_{\mathcal{M}}(U) - |U|$ (\emph{e.g.}, see~\cite{Schrijver2003}). Hence we can find the largest density $\rho^*$ and the associated largest densest subset in polynomial time using binary search: for some value $\rho$, if $U_{\rho} = \emptyset$ then it means that $\rho^* < \rho$, and if $U_{\rho} \neq \emptyset$ it means that $\rho^* \geq \rho$. The exact value of $\rho^*$ can be found as densities can only be rational numbers with denominators bounded by the rank $k$ of the matroid. For the largest densest subset $U_{\rho^*}$, we have $f_{\rho^*}(U_{\rho^*}) = 0$, and when $\rho < \rho^*$ we have $f_{\rho}(U_{\rho}) \leq f_{\rho}(U_{\rho^*}) < 0$.
 
    Although the above procedure can be costly in running time, for some simple matroids that may be of more practical importance, such as laminar or transversal matroids, it should be possible to compute the density-based decomposition faster, because of their particular structures. Moreover, in our algorithms, as we frequently update the ground set on which we compute the decomposition by adding or removing one element, there may be room to improve our time complexity: we leave as an open question whether updating a density-based decomposition when performing these kinds of operations can be done more efficiently, without re-computing the whole decomposition each time. 
    
    Analysing more carefully how the density-based decomposition and the DCS could be updated efficiently may also lead to an application of DCSes to dynamic matroid intersection (note that the EDCS was originally proposed for dynamic graph matching~\cite{BernsteinS15}). In that setting, elements are added into or removed from the ground set and the objective is to maintain an approximate maximum matroid intersection, while guaranteeing a small update time.

	\paragraph*{Related Work} 
	
    Matroid intersection is an ubiquitous subject in theoretical computer science. We refer the reader to the comprehensive book of Schrijver~\cite{Schrijver2003}. Although in the traditional offline setting we know since the 70s that the problem can be solved in polynomial time~\cite{edmonds1970submodular, Edmonds1971, Edmonds1979}, improving the running time of matroid intersection is still a very active area~\cite{BlikstadMNT23,Blikstad21,ChakrabartyLS0W19}.
    
    The importance of matroid intersection comes from the large variety of combinatorial optimization problems it captures, the most well-known being bipartite matching and packing of spanning trees/arborescences. Moreover, other applications can be found in electric circuit theory~\cite{Murota1999matrices,Recski1989matroid}, rigidity theory~\cite{Recski1989matroid}, and network coding~\cite{DoughertyFZ11}. In general, matroids generalize numerous combinatorial constraints; as a result matroid intersection can appear in very diverse contexts. For instance, a recent trend in machine learning is the ``fairness'' constraints (\emph{e.g.}, see~\cite{Chierichetti0LV19} and references therein), which can be encoded by partition or laminar matroids (for nested constraints). Machine scheduling constraints is another example of matroid application, in that case using transversal matroids, see~\cite{GabowT84,XuG94}.
	
	For the one-way communication problem~\cite{kushilevitz1997communication}, the case of maximum matching has been studied in~\cite{AssadiB19, GoelKK12}, for which a $3/2 + \varepsilon$ approximation is obtained. We are not aware of any previous result for the matroid intersection problem in that model. In general, one-way communication is often used to get a better understanding of streaming problems, see~\cite{FeldmanNSZ20, GoelKK12}.
	
	In the \emph{adversarial} streaming, the trivial greedy algorithm building a maximal independent set (an independent set that cannot be extended) achieves a $2$ approximation~\cite{Calinescu2007,Mestre2006}. Improving that approximation ratio is a major open question in the field of streaming algorithms, even for the simple case of bipartite matching (an intersection of two partition matroids). On the hardness side, we know that an approximation ratio better than $1 + \ln 2 \approx 1.69$ cannot be achieved~\cite{Kapralov2021} (previously, an inapproximability of $1 + 1/(e-1) \approx 1.58$ had been established in~\cite{Kapralov13}) for the maximum bipartite matching --- hence for the matroid intersection as well. Note that matroid intersection has been studied in the streaming setting under the adversarial model (in the more general case of weighted/submodular optimisation), for instance see~\cite{ChekuriGQ15, DBLP:conf/nips/FeldmanK018, GJS2021}.
	
	In comparison with the adversarial model, for the \emph{random-order} streaming, Guruganesh and Singla have obtained a $2 - \delta$ approximation ratio (for some small $\delta > 0$) for matroid intersection~\cite{GuruganeshS17}. 
	To our knowledge, it is the only result beating the factor of $2$ for the general matroid intersection problem. 
	In the maximum matching problem (not necessarily in bipartite graphs), a pioneering result was first obtained by Konrad, Magniez, and Mathieu~\cite{KonradMM12} with an approximation ratio strictly below $2$ for simple matchings. The approximation ratio was later improved in a sequence of papers~\cite{GamlathSSS2019,Konrad18,FarhadiHMRR20,bernstein:LIPIcs:2020:12419}. Currently the best result for matchings is due to Assadi and Behnezhad~\cite{AssadiB21}, who obtained the ratio of $3/2 - \delta$ for some small constant $\delta \sim 10^{-14}$. 
    
\section{Density-Based Decomposition}
    \label{sec:density-decomposition}

    Let $\mathcal{M} = (V, \mathcal{I})$ be a matroid on the ground set $V$. Recall that a pair $\mathcal{M} = (V, \mathcal{I})$ is a matroid if the following three conditions hold: (1) $\emptyset \in \mathcal{I}$, (2) if $X\subseteq Y \in \mathcal{I}$, then $X\in \mathcal{I}$, and (3) if $X, Y \in \mathcal{I}, |Y| > |X|$, there exists an element $e \in Y \backslash X$ so that $X \cup \{e\} \in \mathcal{I}$. 
    The sets in $\mathcal{I} \subseteq \mathcal{P}(V)$ are the \emph{independent sets}.
    The \emph{rank} of a subset $X \subseteq V$ is $\rnk_{\mathcal{M}}(X) = \max_{Y \subseteq X,\,Y \in \mathcal{I}}|Y|$. The rank of a matroid is $\rnk_{\mathcal{M}}(V)$. Observe that this notion generalizes that of linear independence in vector spaces.
    
    A subset $C \subseteq V$ is a \emph{circuit} if $C$ is a minimal 
    non-independent set, \emph{i.e.}, for every $v \in C$, $C \backslash \{v\} \in \mathcal{I}$. 
    We will assume that no element in $V$ is a circuit by itself (called ``loop'' in the literature) throughout the paper.  
    The \emph{span} of a subset $X \subseteq V$ in the matroid $\mathcal{M}$ is defined as $\spn_{\mathcal{M}}(X) = \{x \in V,\,\rnk_{\mathcal{M}} (X \cup \{x\}) = \rnk_{\mathcal{M}}(X)\}$, these elements are called spanned by $X$ in $\mathcal{M}$. 
    For more details about matroids, we refer the reader to~\cite{Schrijver2003}.
    
    The \emph{restriction} and \emph{contraction} of a matroid results in another matroid.
    
    \begin{definition}[Restriction]
         Let $\mathcal{M} = (V, \mathcal{I})$ be a matroid, and let $V' \subseteq V$ be a subset. Then we define the restriction of $\mathcal{M}$ to $V'$ as $\mathcal{M}' = \mathcal{M}|V' = (V', \mathcal{I}')$ where $\mathcal{I'} = \{S \subseteq V' : S \in \mathcal{I}\}$.
    \end{definition}
    
    \begin{definition}[Contraction]
        Let $\mathcal{M} = (V, \mathcal{I})$ be a matroid, and let $U$ be a subset of $V$. Then we define the contracted matroid $\mathcal{M}/U = (V \backslash U, \mathcal{I}_U)$ so that, given a maximum independent subset $\mathcal{B}_U$ of $U$, $\mathcal{I}_U = \{S \subseteq V \backslash U : S \cup \mathcal{B}_U \in \mathcal{I}\}$.
    \end{definition}
    
    It is well-known that any choice of $\mathcal{B}_U$  produces the same $\mathcal{I}_U$, as a result the definition of contraction is unambiguous. 
    The following proposition comes directly from the definition. 
	\begin{proposition} \label{prop:rank-contraction}
		Let $\mathcal{M} = (V, \mathcal{I})$ be a matroid and let $A \subseteq B \subseteq V$. Then we have $\rnk_{\mathcal{M}/A}(B\backslash A)  = \rnk_{\mathcal{M}}(B) - \rnk_{\mathcal{M}}(A)$.
	\end{proposition}
    
    Here we recall the definition of density that we will use in the following.
    \densitydefinition*
    
    The following proposition, which we will use frequently, states how the density is changed after a matroid is contracted.
    \begin{proposition}
        \label{prop:subset_increased_rho}
        Let $\mathcal{M} = (V, \mathcal{I})$ be a matroid. If $A \subseteq B \subseteq V$ and $U \subseteq V \backslash B$ we have the following inequality:
        \[\rho_{\mathcal{M}/A}(U) \leq \rho_{\mathcal{M}/B}(U),\]
        assuming that $\rho_{\mathcal{M}/A}(U) < +\infty$.
    \end{proposition}
    
    \begin{proof}
        In fact, $\rnk_{\mathcal{M}/A}(U) \geq \rnk_{\mathcal{M}/B}(U)$, while the cardinality $|U|$ remains obviously the same.
    \end{proof}
    
    The notions of density and matroid contraction allow to define the \emph{density-based decomposition} $U_1,\dots, U_k$ of a subset $V' \subseteq V$ as follows. First, consider the matroid $\mathcal{M}'$ defined as matroid $\mathcal{M}=(V, \mathcal{I})$ restricted to the subset $V' \subseteq V$. Then select from $V'$ the set $U_1$ of largest density in $\mathcal{M}'$ (in case several sets have the same largest density, choose the one with the largest cardinality). Then again choose  the set $U_2$ of largest density (again choose the one with the largest cardinality) in $\mathcal{M}'/U_1$ and so on (see a formal description in Algorithm~\ref{algo:decomposition}). As the rank of the matroid is $k$, after at most $k$ steps in the loop the set $\bigcup_{i = 1}^{k}U_i$ is equal to $V'$. Observe that some latter sets of the decomposition may be empty. Moreover, this decomposition is unique as the choice of maximum cardinality densest subset at each step is unique (see Proposition~\ref{prop:unique-densest}). We note that as we assume that no element is a circuit by itself, our construction guarantees that no set $U_i$ has infinite density. 
    
    \begin{algorithm}[h]
	\caption{Algorithm for building a density-based decomposition of a set $V'$ in  $\mathcal{M}'=(V', \mathcal{I}')$}\label{algo:decomposition}
	\begin{algorithmic}[1]
	\State $\forall\, 1 \leq i \leq k, U_i \gets \emptyset$
	\For{$j = 1 \dots k$}
	    \State $U_j \gets$ the densest subset of largest cardinality in $\mathcal{M}'/(\bigcup_{i = 1}^{j-1}U_i)$ 
	\EndFor
	\end{algorithmic}
	\end{algorithm}
	
	To give some intuition about this decomposition, we provide an example for a laminar matroid, that is represented in Figure~\ref{fig:laminar-matroid} and decomposed in Figure~\ref{fig:laminar-matroid-decomposition}.
	
	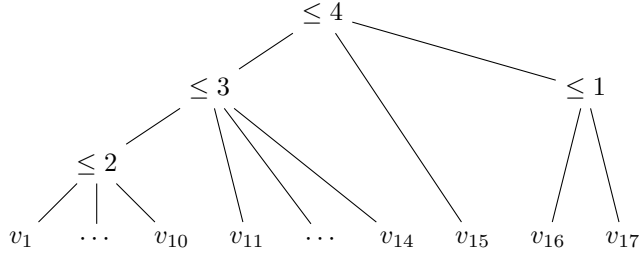
\begin{figure}
    	\centering
        \begin{tikzpicture}
            \node (V1) at (1,0) {$v_1$};
            \node (V2) at (2,0) {$\cdots$};
            \node (V3) at (3,0) {$v_{10}$};
            \node (V4) at (4,0) {$v_{11}$};
            \node (V5) at (5,0) {$\cdots$};
            \node (V6) at (6,0) {$v_{14}$};
            \node (V7) at (7,0) {$v_{15}$};
            \node (V8) at (8,0) {$v_{16}$};
            \node (V9) at (9,0) {$v_{17}$};
            
            \node (X3) at (2,1) {$\leq 2$};
            \draw [-] (V1) -- (X3);
            \draw [-] (V2) -- (X3);
            \draw [-] (V3) -- (X3);
            
            \node (X2) at (3.5,2) {$\leq 3$};
            \draw [-] (X3) -- (X2);
            \draw [-] (V4) -- (X2);
            \draw [-] (V5) -- (X2);
            \draw [-] (V6) -- (X2);
            
            \node (X4) at (8.5,2) {$\leq 1$};
            \draw [-] (V8) -- (X4);
            \draw [-] (V9) -- (X4);
            
            \node (X1) at (5,3) {$\leq 4$};
            \draw [-] (X2) -- (X1);
            \draw [-] (X4) -- (X1);
            \draw [-] (V7) -- (X1);
        \end{tikzpicture}
        \caption{Representation of a laminar matroid $\mathcal{M} = (V, \mathcal{I})$ on a ground set $V = \{v_1, \dots, v_{17}\}$. The leaves represent elements of the ground set, and the inner nodes represent cardinality constraints on the elements in their associated subtree (\emph{e.g.}, if $S \in \mathcal{I}$, then $|S \cap \{v_1, \dots, v_{14}\}| \leq 3$).}
        \label{fig:laminar-matroid}
    \end{figure}
    
    \begin{figure}
    	\centering
        \begin{tikzpicture}
            \node (V1) at (1,0) {$v_1$};
            \node (V2) at (2,0) {$\cdots$};
            \node (V3) at (3,0) {$v_{10}$};
            \node (V4) at (4,0) {$v_{11}$};
            \node (V5) at (5,0) {$\cdots$};
            \node (V6) at (6,0) {$v_{14}$};
            \node (V7) at (7,0) {$v_{15}$};
            \node (V8) at (8,0) {$v_{16}$};
            \node (V9) at (9,0) {$v_{17}$};
            
            \node [draw, red, rectangle] (X3) at (2,1) {$\leq 2$};
            \draw [-] (V1) -- (X3);
            \draw [-] (V2) -- (X3);
            \draw [-] (V3) -- (X3);
            
            \node [draw, red, rectangle] (X2) at (3.5,2) {$\leq 1$};
            \draw [dashed] (X3) -- (X2);
            \draw [-] (V4) -- (X2);
            \draw [-] (V5) -- (X2);
            \draw [-] (V6) -- (X2);
            
            \node (X4) at (8.5,2) {$\leq 1$};
            \draw [-] (V8) -- (X4);
            \draw [-] (V9) -- (X4);
            
            \node [draw, red, rectangle] (X1) at (5,3) {$\leq 1$};
            \draw [dashed] (X2) -- (X1);
            \draw [-] (X4) -- (X1);
            \draw [-] (V7) -- (X1);
        \end{tikzpicture}
        \caption{Density-based decomposition of the laminar matroid $\mathcal{M}$ represented in Figure~\ref{fig:laminar-matroid}. We have the densest subset $U_1 = \{v_1, \dots, v_{10}\}$, then the second densest subset $U_2 = \{v_{11}, \dots, v_{14}\}$ and finally $U_3 = \{v_{15}, v_{16}, v_{17}\}$. Their densities are respectively $5$, $4$, and $3$. Note that here $k = 4$ so we have an additional set $U_4 = \emptyset$ of density zero.}
        \label{fig:laminar-matroid-decomposition}
    \end{figure}
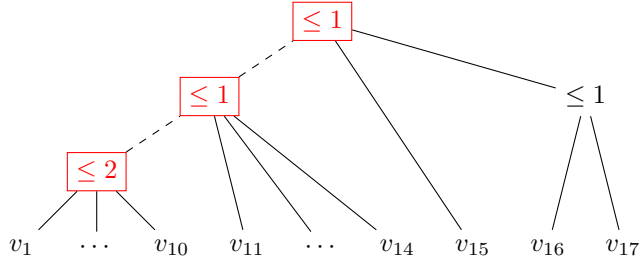

	\begin{proposition}\label{prop:largerdensity}
	    Let $\mathcal{M} = (V, \mathcal{I})$ be a matroid and let $B$ be the subset that reaches the maximum density $\rho^* < +\infty$. Then given any $A \subsetneq B$, $\rho_{\mathcal{M}/ A}(B \backslash A) \geq \rho^{*}$.
	\end{proposition}
	
	\begin{proof}
	    If $\rnk_{\mathcal{M}/A}(B \backslash A) = 0$ then $\rho_{\mathcal{M}/A}(B \backslash A) = +\infty$ and we are done; otherwise, by  Proposition~\ref{prop:rank-contraction}:
	    \[\rho_{\mathcal{M}}(B) = \frac{\rnk_{\mathcal{M}}(A) \cdot \rho_{\mathcal{M}}(A) + \rnk_{\mathcal{M}/A}(B \backslash A) \cdot \rho_{\mathcal{M}/A}(B \backslash A)}{\rnk_{\mathcal{M}}(A) + \rnk_{\mathcal{M}/A}(B \backslash A)},\]
	    hence $\rho_{\mathcal{M}}(B)$ is a weighted average of $\rho_{\mathcal{M}}(A)$ and $\rho_{\mathcal{M}/A}(B \backslash A)$. As $\rho_{\mathcal{M}}(A) \leq \rho^*$ (by definition of $\rho^*$), it implies that $\rho_{\mathcal{M}/A}(B \backslash A) \geq \rho^*$.
	\end{proof}
	
	The following proposition states that the densest sets are closed under union, hence we have the unicity of the maximum cardinality densest subset. 
	
	\begin{proposition} \label{prop:unique-densest}
        Let $\mathcal{M} = (V, \mathcal{I})$ be a matroid. Let $\rho^* = \max_{U \subseteq V}\rho_{\mathcal{M}}(U) < +\infty$. Then given any two sets $W_1$, $W_2$ of density $\rho^*$, $\rho_{\mathcal{M}}(W_1 \cup W_2) = \rho^*$.
    \end{proposition}
    
    \begin{proof}
        If $W_1 \subseteq W_2$, then the proposition is trivially true. So assume that $W_1 \backslash W_2 \neq \emptyset$, and we can
        observe that
        \[\rho^* \leq  \rho_{\mathcal{M}/(W_1 \cap W_2)}(W_1 \backslash (W_1 \cap W_2)) \leq \rho_{\mathcal{M}/W_2}(W_1 \backslash (W_1 \cap W_2)),\]
        where the first inequality uses Proposition~\ref{prop:largerdensity} and the second uses Proposition~\ref{prop:subset_increased_rho}. 
        As a result, by the facts that $\rho_{\mathcal{M}}(W_2) = \rho^*$ and that  $\rho_{\mathcal{M}/W_2}(W_1 \backslash (W_1 \cap W_2)) \geq \rho^*$, we obtain $\rho_{\mathcal{M}}(W_1 \cup W_2) \geq \rho^*$. Hence we have $\rho_{\mathcal{M}}(W_1 \cup W_2) = \rho^*$.
    \end{proof}

    Here is a first proposition about density-based decompositions, stating that the densities decrease (as we could observe in the example of Figure~\ref{fig:laminar-matroid-decomposition}).
	\begin{proposition} \label{prop:non-increasing}
	    For all $1 \leq j \leq k - 1$, $\rho_{\mathcal{M}'/(\bigcup_{i = 1}^{j-1}U_i)}(U_j) \geq \rho_{\mathcal{M}'/(\bigcup_{i = 1}^{j}U_i)}(U_{j+1})$. Moreover, if we have $\rho_{\mathcal{M}'/(\bigcup_{i = 1}^{j-1}U_i)}(U_j) > 0$, then $\rho_{\mathcal{M}'/(\bigcup_{i = 1}^{j-1}U_i)}(U_j) > \rho_{\mathcal{M}'/(\bigcup_{i = 1}^{j}U_i)}(U_{j+1})$.
	\end{proposition}
	
	\begin{proof}
	    We proceed by contradiction. Suppose that $\rho_{\mathcal{M}'/(\bigcup_{i = 1}^{j-1}U_i)}(U_j) < \rho_{\mathcal{M}'/(\bigcup_{i = 1}^{j}U_i)}(U_{j+1})$. Then it implies that $\rho_{\mathcal{M}'/(\bigcup_{i = 1}^{j-1}U_i)}(U_j \cup U_{j+1}) > \rho_{\mathcal{M}'/(\bigcup_{i = 1}^{j-1}U_i)}(U_j)$. Specifically, denoting $k_j = \rnk_{\mathcal{M}'/(\bigcup_{i = 1}^{j-1}U_i)}(U_j)$ and $k_{j+1} = \rnk_{\mathcal{M}'/(\bigcup_{i = 1}^{j}U_i)}(U_{j+1})$, we have $\rho_{\mathcal{M}'/(\bigcup_{i = 1}^{j-1}U_i)}(U_j \cup U_{j+1}) = \rho_{\mathcal{M}'/(\bigcup_{i = 1}^{j-1}U_i)}(U_j) \cdot \frac{k_j}{k_j + k_{j+1}} + \rho_{\mathcal{M}'/(\bigcup_{i = 1}^{j}U_i)}(U_{j+1}) \cdot \frac{k_{j+1}}{k_j + k_{j+1}} > 
	    \rho_{\mathcal{M}'/(\bigcup_{i = 1}^{j-1}U_i)}(U_j)$, 
	    contradicting the hypothesis that $U_j$ was the densest set in $\mathcal{M}'/(\bigcup_{i = 1}^{j-1}U_i)$.
	    
	    For the second part of the proposition, suppose that $\rho_{\mathcal{M}'/(\bigcup_{i = 1}^{j-1}U_i)}(U_j) > 0$ and $\rho_{\mathcal{M}'/(\bigcup_{i = 1}^{j-1}U_i)}(U_j) = \rho_{\mathcal{M}'/(\bigcup_{i = 1}^{j}U_i)}(U_{j+1})$. Then it implies that $\rho_{\mathcal{M}'/(\bigcup_{i = 1}^{j-1}U_i)}(U_j \cup U_{j+1}) = \rho_{\mathcal{M}'/(\bigcup_{i = 1}^{j-1}U_i)}(U_j)$, contradicting the supposition that $U_j$ was the maximum cardinality densest set. 
	\end{proof}
    
    The following proposition comes straightforwardly from the definition of the densities:
	\begin{proposition} \label{prop:sum-density}
	    We always have $\sum_{j = 1}^k \rnk_{\mathcal{M}'/(\bigcup_{i = 1}^{j-1}U_i)}(U_j) \cdot \rho_{\mathcal{M}'/(\bigcup_{i = 1}^{j-1}U_i)}(U_j) = |V'|$.
	\end{proposition}
	
	Now we define the \emph{associated density} of a given element $v \in V$ with respect to the decomposition of $V'$.
	
	\begin{definition}
	    \label{def:tilderho}
	    Let $U_1,\dots, U_k$ to the density-based decomposition of $V'$. Then, given an element $v \in V$, its \emph{associated density} with respect to the decomposition of $V'$ is defined as
	    \[
	        \tilde{\rho}_{\mathcal{M}}(v) = \left\{
            \begin{array}{ll}
                \rho_{\mathcal{M}'/(\bigcup_{i = 1}^{j-1}U_i)}(U_j)\text{ for $j = \min\{j \in \llbracket 1, k \rrbracket : v \in \spn_{\mathcal{M}}(\bigcup_{i = 1}^{j}U_i)\}$}& \text{if } v \in \spn_{\mathcal{M}}(V') \\
                0 & \text{otherwise}
            \end{array}
        \right.
	    \]
	\end{definition}
	We emphasize that the associated density $\tilde{\rho}_{\mathcal{M}}$ is defined for \emph{all} elements in $V$, not just the elements of $V'$ (this is why we use the subscript $\mathcal{M}$ instead of $\mathcal{M}'$). We also emphasize that 
	here the associated density is dependent on $V'$, even though that dependence is not displayed in our notation:
	we will just write $\tilde{\rho}_{\mathcal{M}}$, instead of the more cumbersome $\tilde{\rho}_{\mathcal{M},V'}$. 
	For elements $v \in V'$, note that if $v \in U_j$ then we have necessarily $\tilde{\rho}_{\mathcal{M}}(v) = \rho_{\mathcal{M}'/(\bigcup_{i = 1}^{j-1}U_i)}(U_j)$; in fact, if $v$ is spanned by $\bigcup_{i = 1}^{j_0}U_i$ for some $j_0 < j$, then we could have increased the density of $U_{j_0}$ by adding $v$ into $U_{j_0}$, contradicting the assumption that $U_{j_0}$ was the densest subset when it was selected.
	
	We now explain how such a decomposition behaves when an element is added to or   deleted from the set $V'$. These two following lemmas are crucial in the existence proof of DCSes. Their statements are quite natural (for instance, stating that adding an element does not cause a diminution of the density associated with any other element, and cannot increase the density of that new element by more than one), however their proofs are rather technical and in fact proving these lemmas is the most difficult step to show the existence of DCSes. From now on, we will use the exponents $^{\textrm{\normalfont old}}$ and $^{\textrm{\normalfont new}}$ to denote the states before and after the insertion/deletion operation.
	
	The proofs of the following lemmas can be found in Appendix~\ref{app:proof}.
	
	\begin{lemma}\label{lem:add-increase}
	    Suppose a new element $u^{\textrm{\normalfont new}} \in V \backslash V'$ is added to $V'$. Then we have the following properties:
	    \begin{enumerate}[(i)]
	        \item For all $j \in \llbracket 1, k \rrbracket$, for all $v \in U_j^{\textrm{\normalfont old}}$, $\tilde{\rho}_{\mathcal{M}}^{\textrm{\normalfont new}}(v) \geq \rho_{\mathcal{M}'^{\textrm{\normalfont old}}/(\bigcup_{i = 1}^{j-1}U_i^{\textrm{\normalfont old}})}(U_j^{\textrm{\normalfont old}})$.
	        \item For all $v \in V$, $\tilde{\rho}_{\mathcal{M}}^{\textrm{\normalfont new}}(v) \geq \tilde{\rho}_{\mathcal{M}}^{\textrm{\normalfont old}}(v)$.
	        \item We have the inequality $\tilde{\rho}_{\mathcal{M}}^{\textrm{\normalfont old}}(u^{\textrm{\normalfont new}}) \leq \tilde{\rho}_{\mathcal{M}}^{\textrm{\normalfont new}}(u^{\textrm{\normalfont new}}) \leq \tilde{\rho}_{\mathcal{M}}^{\textrm{\normalfont old}}(u^{\textrm{\normalfont new}}) + 1$.
	        \item For all $v \in V'$ such that $\tilde{\rho}_{\mathcal{M}}^{\textrm{\normalfont old}}(v) < \tilde{\rho}_{\mathcal{M}}^{\textrm{\normalfont old}}(u^{\textrm{\normalfont new}})$ or $\tilde{\rho}_{\mathcal{M}}^{\textrm{\normalfont old}}(v) > \tilde{\rho}_{\mathcal{M}}^{\textrm{\normalfont old}}(u^{\textrm{\normalfont new}}) + 1$, we have the equality $\tilde{\rho}_{\mathcal{M}}^{\textrm{\normalfont old}}(v) = \tilde{\rho}_{\mathcal{M}}^{\textrm{\normalfont new}}(v)$.
	    \end{enumerate}
	\end{lemma}
	
	\begin{lemma}\label{lem:del-decrease}
	    Suppose an old element $u^{\textrm{\normalfont old}} \in V'$ is deleted from $V'$. Then we have the following properties:
	    \begin{enumerate}[(i)]
	        \item For all $j \in \llbracket 1, k \rrbracket$, for all $v \in U_j^{\textrm{\normalfont old}}$, $\tilde{\rho}_{\mathcal{M}}^{\textrm{\normalfont new}}(v) \leq \rho_{\mathcal{M}'^{\textrm{\normalfont old}}/(\bigcup_{i = 1}^{j-1}U_i^{\textrm{\normalfont old}})}(U_j^{\textrm{\normalfont old}})$.
	        \item For all $v \in V$, $\tilde{\rho}_{\mathcal{M}}^{\textrm{\normalfont new}}(v) \leq \tilde{\rho}_{\mathcal{M}}^{\textrm{\normalfont old}}(v)$.
	        \item We have the inequality $\tilde{\rho}_{\mathcal{M}}^{\textrm{\normalfont old}}(u^{\textrm{\normalfont old}}) \geq \tilde{\rho}_{\mathcal{M}}^{\textrm{\normalfont new}}(u^{\textrm{\normalfont old}}) \geq \tilde{\rho}_{\mathcal{M}}^{\textrm{\normalfont old}}(u^{\textrm{\normalfont old}}) - 1$.
	        \item For all $v \in V'$ such that $\tilde{\rho}_{\mathcal{M}}^{\textrm{\normalfont old}}(v) > \tilde{\rho}_{\mathcal{M}}^{\textrm{\normalfont old}}(u^{\textrm{\normalfont old}})$ or $\tilde{\rho}_{\mathcal{M}}^{\textrm{\normalfont old}}(v) < \tilde{\rho}_{\mathcal{M}}^{\textrm{\normalfont old}}(u^{\textrm{\normalfont old}}) - 1$, we have the equality $\tilde{\rho}_{\mathcal{M}}^{\textrm{\normalfont old}}(v) = \tilde{\rho}_{\mathcal{M}}^{\textrm{\normalfont new}}(v)$.
	    \end{enumerate}
	\end{lemma}

\section{Density-Constrained Subsets for Matroid Intersection}

    Consider two matroids $\mathcal{M}_1 = (V, \mathcal{I}_1)$ and $\mathcal{M}_2 = (V, \mathcal{I}_2)$, both of rank $k$ (if the matroids have different ranks, we can truncate the rank of the matroid of larger rank without changing the solution of the matroid intersection problem). We recall the definition of a \emph{Density-Constrained Subset} (DCS).
    
    \dcsdefinition*
    
    Here is a simple bound on the size of a DCS.
    \begin{proposition} \label{prop:size-bound}
        For any set $V' \subseteq V$ satisfying Property (i) of Definition~\ref{def:dcs}, $|V'| \leq \beta \cdot \mu(V)$, where $\mu(V)$ denotes the maximum cardinality common independent subset in $V$.
    \end{proposition}
    
    \begin{proof} 
        We proceed by contradiction. By Theorem~\ref{thm:matroid-intersection}, we know that there exists a set $S \subseteq V$ such that $\rnk_{\mathcal{M}_1}(S) + \rnk_{\mathcal{M}_2}(V \backslash S) = \mu(V)$. If $|V'| > \beta \cdot \mu(V)$, then it means that either $|V' \cap S| > \beta \cdot \rnk_{\mathcal{M}_1}(S)$ or that $|V' \cap (V \backslash S)| > \beta \cdot \rnk_{\mathcal{M}_2}(V \backslash S)$. In both cases, we have a densest subset, either in $\mathcal{M}'_1$ or $\mathcal{M}'_2$, that has a density larger than $\beta$, contradicting Property (i) of Definition~\ref{def:dcs}.
    \end{proof}
    
    We show the existence of $(\beta, \beta^-)$-DCSes by construction, using a local search algorithm inspired by the one used in~\cite{AssadiB19}. In our proof we introduce a new potential function and we use Lemmas~\ref{lem:add-increase} and~\ref{lem:del-decrease} to generalize their procedure; details of the proof can be found in Appendix~\ref{app:proof}.
    
    \begin{theorem} \label{thm:construction}
        For any two matroids $\mathcal{M}_1 = (V, \mathcal{I}_1)$ and $\mathcal{M}_2 = (V, \mathcal{I}_2)$ of rank $k$, and for any integer parameters $\beta \geq \beta^- + 7$, a $(\beta,\beta^-)$-DCS can be computed using at most $2\cdot \beta^2 \cdot \mu(V)$ local improvement steps.
    \end{theorem}
    
    The main interest of DCS relies in that they always contain a relatively good approximation of the maximum cardinality matroid intersection. 
    \ratiotheorem*
	
	\begin{proof}
		Let $V'$ be a $(\beta, \beta^-)$-DCS, and let $C_1$ and $C_2$ be sets such that $C_1 \cup C_2 = V'$, $C_1 \cap C_2 = \emptyset$ and minimizing the sum $\rnk_{\mathcal{M}'_1}(C_1) + \rnk_{\mathcal{M}'_2}(C_2) = \rnk_{\mathcal{M}_1}(C_1) + \rnk_{\mathcal{M}_2}(C_2)$; by Theorem~\ref{thm:matroid-intersection} we know that $\rnk_{\mathcal{M}'_1}(C_1) + \rnk_{\mathcal{M}'_2}(C_2)= \mu(V')$, the size of the maximum common independent set in $V'$.
		
		Now consider the optimal common independent set $O$ in $V$. Our objective is to bound both $|O \backslash S|$ and $|S|$ for some well-chosen subset $S \subseteq O$ to get an upper bound of $|O|$. We will build that auxiliary set $S$ as follows, starting with $S = \emptyset$. If there exists an element $o_1 \in O$ such that $o_1 \notin \spn_{\mathcal{M}_1}(C_1) \cup  \spn_{\mathcal{M}_2}(C_2)$, then we add $o_1$ into $S$ and we now consider the contracted matroids $\mathcal{M}_1/S$ and $\mathcal{M}_2/S$. We keep the same sets $C_1$ and $C_2$ and we try again to find an element $o_2 \in O \backslash S$ such that $o_2 \notin \spn_{\mathcal{M}_1/S}(C_1) \cup \spn_{\mathcal{M}_2/S}(C_2)$, and we add $o_2$ to $S$. We repeat this operation until it is no longer possible to add into $S$ any other element of $O$. 
		The idea behind this greedy procedure to build $S$ is that, if we instead defined $S$ naively as the set of elements in $O$ that are not in $\spn_{\mathcal{M}_1}(C_1) \cup  \spn_{\mathcal{M}_2}(C_2)$ (which would be a simpler way to get a set $S$ satisfying inequality~(\ref{eq:os-bound}) below), then this may yield a much bigger set $S$ for which we could not get a proper bound, whereas here the greedy procedure gives us a tool to bound $|S|$ as it will allow us to prove the crucial inequality~(\ref{eq:sum-r}) later.
		
		By the above greedy  procedure, $O \backslash S$ is a common independent subset in $\mathcal{M}_1/S$ and $\mathcal{M}_2/S$ restricted to $V' \cup O \backslash S$, and $\spn_{\mathcal{M}_1/S}(C_1) \cup \spn_{\mathcal{M}_2/S}(C_2) \supseteq V' \cup O \backslash S$. 
		We now observe that 
		\begin{align*}
		|O \backslash S| & \leq  \min_{U \subseteq V' \cup O \backslash S} (\rnk_{\mathcal{M}_1/S}(U) + \rnk_{\mathcal{M}_2/S}((V' \cup O \backslash S) \backslash U)) \\
		&\leq \rnk_{\mathcal{M}_1/S}(\spn_{\mathcal{M}_1/S}(C_1)) + \rnk_{\mathcal{M}_2/S}((V' \cup O \backslash S) \backslash (\spn_{\mathcal{M}_1/S}(C_1)))\\
		& \leq \rnk_{\mathcal{M}_1/S}(\spn_{\mathcal{M}_1/S}(C_1)) + \rnk_{\mathcal{M}_2/S}(\spn_{\mathcal{M}_2/S}(C_2)) \\
		& =\rnk_{\mathcal{M}_1/S}(C_1) + \rnk_{\mathcal{M}_2/S}(C_2) \\
		& \leq \rnk_{\mathcal{M}_1}(C_1) + \rnk_{\mathcal{M}_2}(C_2) \\
		&  = \mu(V'),   \label{eq:os-bound}
		\end{align*}
		where in the first inequality we use Theorem~\ref{thm:matroid-intersection}, 
		in the second inequality we consider $U = \spn_{\mathcal{M}_1/S}(C_1)$, 
		in the third inequality we use that $(V' \cup O \backslash S) \backslash (\spn_{\mathcal{M}_1/S}(C_1)) \subseteq \spn_{\mathcal{M}_2/S}(C_2)$, and in the last inequality we use that the rank function in a contracted matroid is always smaller than the rank function in the original matroid. Thereby we obtain
		\begin{equation} \label{eq:os-bound}
		    |O \backslash S| \leq \mu(V').
		\end{equation}
		
		Hence we need to upper-bound the value of $|S|$. Some carefully chosen subsets $R_{l, i}$ and $Q_{l, i}$ will allow us to get that upper-bound, and their construction is displayed in the following lemmas --- it is in the proof of these lemmas that the DCS structure is fully exploited. Observing that $\beta^- \cdot |S|$ is bounded by the sum of the $\tilde{\rho}_{\mathcal{M}_l}(o_i)$ (as for each $o_i \in S$, we have $\beta^- \leq \tilde{\rho}_{\mathcal{M}_1}(o_i) + \tilde{\rho}_{\mathcal{M}_2}(o_i)$, because of Property~(ii) of the DCS), we will build disjoint subsets $R_{l,i}$ of $V'$ (Lemma~\ref{lem:defineR}) to bound each $\tilde{\rho}_{\mathcal{M}_l}(o_j)$ with $|R_{l,j}|$ (in particular, see Lemma~\ref{lem:defineR}~(iv)). We will then use an auxiliary partition $Q_{l,j}$ of the union of the $R_{l,j}$s (Lemma~\ref{lem:defineQ}) to bound the total size of the $R_{l,j}$s, using the properties of the DCS and the properties of those sets. By wrapping-up everything in the end this will allow us to get a bound on the size of $S$, similarly as~\cite{AssadiB19}.
		
		We recall that the sets $U_{l,i}$ refer to the density-based decomposition of $V'$ in the matroid $\mathcal{M}_l$.
		
		\begin{lemma}
		\label{lem:defineR}
		    For $l \in \{1,2\}$, we can build sets $R_{l, 1}, \dots, R_{l, |S|}$ satisfying the following properties:
		    \begin{enumerate}[(i)]
		        \item the $R_{l, i}$ are disjoint;
		        \item for all $j \in \llbracket 1, |S| \rrbracket$ we have $R_{l, j} \subseteq V' \backslash C_l$;
		        \item for all $j \in \llbracket 1, |S| \rrbracket$, for all $v \in R_{l, j}$, $|R_{l, j}| = \lfloor\tilde{\rho}_{\mathcal{M}_l}(v)\rfloor - 1$;
		        \item for all $j \in \llbracket 1, |S| \rrbracket$, $|R_{l, j}| \geq \lfloor \tilde{\rho}_{\mathcal{M}_l}(o_j) \rfloor - 1$.
		    \end{enumerate}
		\end{lemma}
		
		\begin{proof}
	        Fix an $l$. 	
		    We divide $S$ into two groups: those that are spanned by $\bigcup_{i=1}^{k}U_{l,i}$ and those that are not. Precisely, 
		    $S_U = S \cap  \spn_{\mathcal{M}_l}(\bigcup_{i=1}^{k}U_{l,i})$ 
		    and $S_{\overline{U}} = S\backslash S_{U}$. 
		    
		    We will extract from $U_{l,1},\dots, U_{l,k}$ subsets $R_{l,x}$ for each $o_x \in S_U$. 
		    For the other elements $o_y \in S_{\overline{U}}$, we 
		    create $R_{l,y}=\emptyset$ and associate $o_y$ with $R_{l,y}$. It is easy to verify that Properties (ii)-(iv) hold in the latter case (for Property (iv), recall that by Definition~\ref{def:tilderho}, 
		    $\tilde{\rho}_{\mathcal{M}}(o_y)=0$).
		    We next explain how to construct $R_{l,x}$ for $o_x \in S_U$. 
		
	        For $j = 1$ to $k$, we split a subset of $U_{l,j} \backslash C_l$ into \[r_{l, j} = \max\left(0,  \rnk_{\mathcal{M}'_l/(\bigcup_{i = 1}^{j-1}U_{l,i})}(U_{l,j}) - \rnk_{\mathcal{M}'_l/(\bigcup_{i = 1}^{j-1}U_{l,i} \cap C_l)}(U_{l,j} \cap C_l)\right)\] sets of size $\lfloor\rho_{\mathcal{M}'_l/(\bigcup_{i = 1}^{j-1}U_{l,i})}(U_{l,j})\rfloor - 1$. It is always possible as we have, when $r_{l,j} > 0$,
		    \begin{align*}
		        \left\lfloor\frac{|U_{l,j} \backslash C_l|}{r_{l,j}}\right\rfloor
		        &=\left\lfloor\frac{|U_{l,j} \backslash C_l|}{\rnk_{\mathcal{M}'_l/(\bigcup_{i = 1}^{j-1}U_{l,i})}(U_{l,j}) - \rnk_{\mathcal{M}'_l/(\bigcup_{i = 1}^{j-1}U_{l,i} \cap C_l)}(U_{l,j} \cap C_l)}\right\rfloor\\
		        &\geq \left\lfloor\frac{|U_{l,j} \backslash C_l|}{\rnk_{\mathcal{M}'_l/(\bigcup_{i = 1}^{j-1}U_{l,i})}(U_{l,j}) - \rnk_{\mathcal{M}'_l/(\bigcup_{i = 1}^{j-1}U_{l,i})}(U_{l,j} \cap C_l)}\right\rfloor &\\
		        &= \lfloor\rho_{\mathcal{M}'_l/(\bigcup_{i = 1}^{j-1}U_{l,i} \cup (C_l \cap U_{l, j}))}(U_{l,j} \backslash C_l)\rfloor & \text{by Proposition~\ref{prop:rank-contraction}} &\\
		        & \geq \lfloor\rho_{\mathcal{M}'_l/(\bigcup_{i = 1}^{j-1}U_{l,i})}(U_{l,j})\rfloor. & \text{by Proposition~\ref{prop:largerdensity}}
		    \end{align*}
		    where in the first inequality we used that $\rnk_{\mathcal{M}'_l/(\bigcup_{i = 1}^{j-1}U_{l,i})}(U_{l,j} \cap C_l) \leq \rnk_{\mathcal{M}'_l/(\bigcup_{i = 1}^{j-1}U_{l,i} \cap C_l)}(U_{l,j} \cap C_l)$.
		    
		    Then the $R_{l,x}$s, for $o_x \in S_U$ are decided by a greedy procedure. Let $x_1, \dots, x_{|S_U|}$ be the indices of the elements of $S_U$, ordered so that $\tilde{\rho}_{\mathcal{M}_l}(o_{x_1}) \geq \cdots \geq \tilde{\rho}_{\mathcal{M}_l}(o_{x_{|S_u|}})$. The first $r_{l,1}$ subsets drawn from $U_{l,1}\backslash C_l$ are assigned to be $R_{l,x_1}, \dots R_{l,x_{r_{l,1}}}$; the following 
		    $r_{l,2}$ subsets drawn from $U_{l,2}\backslash C_2$ are assigned to be $R_{l,x_{r_{l,1}+1}}, \dots, R_{l, x_{r_{l,1}+r_{l,2}}}$, and so on. 
		    
		    Notice that by this procedure, properties (ii) and (iii) hold easily for $R_{l,x}$, $o_x \in S_U$. To prove property (iv), we will prove the following inequality for all $j$: 
		    
		    \begin{equation} \label{eq:sum-r}
		        \sum_{i = 1}^j r_{l,i} \geq \left|S \cap \spn_{\mathcal{M}_l}\left(\bigcup_{i = 1}^j U_{l,i}\right)\right|.
		    \end{equation}
		    
		    To see why inequality~(\ref{eq:sum-r}) implies (iv), 
		    for $j = 1$ to $k$, let us define the set $S_j$ of elements with ``density level'' $j$, \emph{i.e.}, $S_j = S \cap (\spn_{\mathcal{M}_l}(\bigcup_{i = 1}^j U_{l,i}) \backslash \spn_{\mathcal{M}_l}(\bigcup_{i = 1}^{j-1} U_{l,i}))$. If $o_{x_t} \in S_j$, by Definition~\ref{def:tilderho}, we need to associate $o_{x_t}$ with a set $R_{l,x_t}$ drawn 
		    from one of $U_{l,1}, \dots, U_{l,j}$, as such a set $R_{l,x_t}$ will have a size larger than or equal to $\lfloor\rho_{\mathcal{M}'_l/(\bigcup_{i = 1}^{j-1}U_{l,i})}(U_{l,j})\rfloor - 1 = \lfloor \tilde{\rho}_{\mathcal{M}_l}(o_{x_{t}}) \rfloor - 1$. The majorization in~(\ref{eq:sum-r}) shows that our greedy procedure will guarantee that a large enough set is assigned to $o_{x_t}$, as the inequality~(\ref{eq:sum-r}) implies that $t \leq \sum_{i = 1}^j r_{l,i}$, hence Property~(iv) would follow.
		    
		    To prove~(\ref{eq:sum-r}), we begin by observing that our greedy procedure in constructing $S$ ensures that 
		    \[\rnk_{\mathcal{M}_l}(C_l \cup S) = \rnk_{\mathcal{M}_l}(C_l) + |S|,\]
		    implying that no circuit in $\mathcal{M}_l$ involves a non-empty subset of $S$ and a non-empty subset of a base in $C_l$. Therefore, given any $\hat{C}_l \subseteq C_l$ and 
		    $\hat{S} \subseteq S$, 
		    
		    \[\rnk_{\mathcal{M}_l}(\hat{C}_l \cup \hat{S}) = \rnk_{\mathcal{M}_l}(\hat{C}_l) + |\hat{S}|.\]

		    With this observation, we can derive
		    \begin{align*}
		        \rnk_{\mathcal{M}_l}\left(\bigcup_{i = 1}^j U_{l,i}\right) &= \rnk_{\mathcal{M}_l}\left(\spn_{\mathcal{M}_l}\left(\bigcup_{i = 1}^j U_{l,i}\right)\right)\\
		        &\geq \rnk_{\mathcal{M}_l}\left((C_l \cup S) \cap \spn_{\mathcal{M}_l}\left(\bigcup_{i = 1}^j U_{l,i}\right)\right)\\
		        &= \rnk_{\mathcal{M}_l}\left(C_l \cap \spn_{\mathcal{M}_l}\left(\bigcup_{i = 1}^j U_{l,i}\right)\right) + \left|S \cap \spn_{\mathcal{M}_l}\left(\bigcup_{i = 1}^j U_{l,i}\right)\right|\\
		        &\geq \rnk_{\mathcal{M}_l}\left(C_l \cap \bigcup_{i = 1}^j U_{l,i}\right) + \left|S \cap \spn_{\mathcal{M}_l}\left(\bigcup_{i = 1}^j U_{l,i}\right)\right|,
		    \end{align*}
		    where the last inequality is actually an equality, as we have $C_l \cap \bigcup_{i = 1}^j U_{l,i} = C_l \cap \spn_{\mathcal{M}_l}\left(\bigcup_{i = 1}^j U_{l,i}\right)$ here.\footnote{However, for Lemmas~\ref{lem:communication-dcs} and~\ref{lem:underfull-intersection} in the next two sections, this will be in fact an inequality, as in the proof of those two lemmas, $C_l$ may contain elements not in $V'$.}

		    We now finish the proof of inequality~(\ref{eq:sum-r}) by observing that
		    \begin{multline*}
		        \rnk_{\mathcal{M}_l}\left(\bigcup_{i = 1}^j U_{l,i}\right)- \rnk_{\mathcal{M}_l}\left(C_l \cap \bigcup_{i = 1}^j U_{l,i}\right)\\ 
		        = \sum_{i = 1}^{j} \left(\rnk_{\mathcal{M}'_l/(\bigcup_{i = 1}^{j-1}U_{l,i})}(U_{l,j}) - \rnk_{\mathcal{M}'_l/(\bigcup_{i = 1}^{j-1}U_{l,i} \cap C_l)}(U_{l,j} \cap C_l)\right) \leq \sum_{i=1}^{j} r_{l,i},
		    \end{multline*}
		    where the equality comes from applying Proposition~\ref{prop:rank-contraction} recursively.
		    
		    We have by now shown that Properties (ii)-(iv) hold in general. Property (i) holds trivially by our construction. Thus the proof is complete. 
		\end{proof}
		
		We denote $R_l = \bigcup_{i=1}^{|S|}R_{l,i}$ and 
		$R = \bigcup_{l \in \{1,2\}}R_l$. Note that 
		$R_l \subseteq V' \backslash C_l$ and $R   \subseteq V'$.
		
		\begin{lemma}
		  \label{lem:defineQ}
		    For $l \in \{1,2\}$, we can build sets $Q_{l, 1}, \dots, Q_{l, \rnk_{\mathcal{M}_l}(R_{3-l})}$ satisfying the following properties:
		    \begin{enumerate}[(i)]
		        \item the $Q_{l, j}$ are disjoint;
		        \item $\bigcup_{i = 1}^{\rnk_{\mathcal{M}_l}(R_{3-l})} Q_{l,i} = R_{3-l}$;
		        \item for all $v \in Q_{l,i}$, $|Q_{l,i}| \leq \tilde{\rho}_{\mathcal{M}_l}(v) + 1$.
		    \end{enumerate}
		\end{lemma}
		
		\begin{proof}
		    Fix an $l$. For $j = 1$ to $k$, we split the set $U_{l,j} \cap R_{3-l}$ into $\rnk_{\mathcal{M}'_l/(\bigcup_{i = 1}^{j-1}U_{l,i} \cap R_{3-l})}(U_{l,j} \cap R_{3-l})$ sets of size at most
		    \begin{align*}
		        \left\lceil \frac{|U_{l,j} \cap R_{3-l}|}{\rnk_{\mathcal{M}'_l/(\bigcup_{i = 1}^{j-1}U_{l,i} \cap R_{3-l})}(U_{l,j} \cap R_{3-l})} \right\rceil &\leq \rho_{\mathcal{M}'_l/(\bigcup_{i = 1}^{j-1}U_{l,i} \cap R_{3-l})}(U_{l,j} \cap R_{3-l}) + 1\\
		        &\leq \rho_{\mathcal{M}'_l/(\bigcup_{i = 1}^{j-1}U_{l,i})}(U_{l,j} \cap R_{3-l}) + 1 &\text{by Proposition~\ref{prop:subset_increased_rho}}\\
		        &\leq \rho_{\mathcal{M}'_l/(\bigcup_{i = 1}^{j-1}U_{l,i})}(U_{l,j}) + 1. &\text{by construction of $U_{l,j}$}
		    \end{align*}
		    These are the aforementioned sets $Q_{l,x}$.
		    It is clear that those $Q_{l,x}$ will be disjoint, and that for all $v \in Q_{l,x} \subseteq U_{l,j}$, we have \[\tilde{\rho}_{\mathcal{M}_l}(v) + 1 = \rho_{\mathcal{M}'_l/(\bigcup_{i = 1}^{j-1}U_{l,i})}(U_{l,j}) + 1 \geq |Q_{l,x}|.\]
		    Observe that by induction, for any $1 \leq r \leq k$, we have $\sum_{j = 1}^r \rnk_{\mathcal{M}'_l/(\bigcup_{i = 1}^{j-1}U_{l,i} \cap R_{3-l})}(U_{l,j} \cap R_{3-l}) = \rnk_{\mathcal{M}'_l}(R_{3 - l} \cap (\bigcup_{i = 1}^{r}U_{l,i}))$ (using Proposition~\ref{prop:rank-contraction}) and hence for $r = k$ we get \[\sum_{j = 1}^k \rnk_{\mathcal{M}'_l/(\bigcup_{i = 1}^{j-1}U_{l,i} \cap R_{3-l})}(U_{l,j} \cap R_{3-l}) = \rnk_{\mathcal{M}'_l}(R_{3 - l}),\] therefore the number of sets $R_{l,x}$ built that way is exactly $\rnk_{\mathcal{M}'_l}(R_{3 - l})$, as desired.
		\end{proof}
		
		We now continue the proof of Theorem~\ref{thm:dcs-ratio}.
		For all $v \in R \subseteq V'$, we know by Property~(i) of Definition~\ref{def:dcs} that:
		\[\tilde{\rho}_{\mathcal{M}_1}(v) + \tilde{\rho}_{\mathcal{M}_2}(v) \leq \beta.\]
		Hence summing over all the elements of $R$:
		\begin{align*}
		    \beta \cdot |R| &\geq \sum_{v \in R}\tilde{\rho}_{\mathcal{M}_1}(v) + \tilde{\rho}_{\mathcal{M}_2}(v) \\
		    &= \mathop{\sum_{l \in \{1,2\}}}_{i \in \{1, \dots |S|\}}\sum_{v \in R_{l,i}}\tilde{\rho}_{\mathcal{M}_l}(v) + \mathop{\sum_{l \in \{1,2\}}}_{i \in \{1, \dots \rnk_{\mathcal{M}_l}(R_{3-l})\}}\sum_{v \in Q_{l,i}}\tilde{\rho}_{\mathcal{M}_l}(v)\\
		    &\geq \mathop{\sum_{l \in \{1,2\}}}_{i \in \{1, \dots |S|\}} |R_{l,i}| \cdot (|R_{l,i}| + 1) + \mathop{\sum_{l \in \{1,2\}}}_{i \in \{1, \dots \rnk_{\mathcal{M}_l}(R_{3-l})\}} |Q_{l, i}| \cdot (|Q_{l,i}| - 1)\\
		    &= \mathop{\sum_{l \in \{1,2\}}}_{i \in \{1, \dots |S|\}} |R_{l,i}|^2 + \mathop{\sum_{l \in \{1,2\}}}_{i \in \{1, \dots \rnk_{\mathcal{M}_l}(R_{3-l})\}} |Q_{l, i}|^2\\
		    &\geq \frac{|R|^2}{2\cdot |S|} + \frac{|R|^2}{\rnk_{\mathcal{M}_1}(R_{2}) + \rnk_{\mathcal{M}_2}(R_{1})}.
		\end{align*}
		To move from the first to the second line we use that for any element in $R_l$, that element also appears in $Q_{3 - l}$, so that we can get in the sum both terms $\tilde{\rho}_{\mathcal{M}_l}(v)$ and $\tilde{\rho}_{\mathcal{M}_{3-l}}(v)$. 
		The first inequality uses Lemmas~\ref{lem:defineR}~(iii) and~\ref{lem:defineQ}~(iii). To move from the third to the fourth line we use that $\sum_{l,i} |R_{l,i}| = \sum_{l,i} |Q_{l, i}| = |R|$.
		The last inequality comes from the minimization of the function under the constraint $\sum_{l,i} |R_{l,i}| = \sum_{l,i} |Q_{l, i}| = |R|$. 
		
		Hence we get
		\[\beta \geq  \frac{|R|}{2\cdot |S|} + \frac{|R|}{\rnk_{\mathcal{M}_1}(R_{2}) + \rnk_{\mathcal{M}_2}(R_{1})}.\] 
		
		As the elements of $S$ satisfy Property~(ii) of Definition~\ref{def:dcs}, and because of Property~(iv) in Lemma~\ref{lem:defineR}, we know that for all $o_i \in S$,
		\[\beta^- \leq \tilde{\rho}_{\mathcal{M}_1}(o_i) + \tilde{\rho}_{\mathcal{M}_2}(o_i) \leq |R_{1,i}| + |R_{2,i}| + 4,\]
		so by averaging over all the elements of $S$ we get
		\[\beta^- \leq \frac{|R|}{|S|} + 4.\]
		
		Therefore we finally obtain
		\[\left(\beta - \frac{\beta^- - 4}{2}\right) \cdot (\rnk_{\mathcal{M}_1}(R_{2}) + \rnk_{\mathcal{M}_2}(R_{1})) \geq |R|.\]
		Then, as $(\beta^- - 4) \cdot |S| \leq |R|$ and $\rnk_{\mathcal{M}_1}(R_{2}) + \rnk_{\mathcal{M}_2}(R_{1}) \leq \rnk_{\mathcal{M}_1}(C_1) + \rnk_{\mathcal{M}_2}(C_2) = \mu(V')$ (because $R_{l-2} \subseteq C_l$) we finally have $\left(\beta - \frac{\beta^- - 4}{2}\right) \cdot \mu(V') \geq (\beta^- - 4) \cdot |S|$. Now using~(\ref{eq:os-bound}), we obtain:
		\[\mu(V) = |O \backslash S| + |S| \leq \left(1 + \frac{\beta}{\beta^--4} - \frac{1}{2}\right) \cdot \mu(V') = \left(\frac{1}{2} + \frac{\beta}{\beta^--4}\right) \cdot \mu(V') \leq \left(\frac{3}{2} + \varepsilon\right) \cdot \mu(V'),\]
		as $(\beta^- - 4) \cdot (1 + \varepsilon) \geq \beta$. This concludes the proof.
	\end{proof}
	
	\begin{remark} \label{rmk:val-eps}
	    We can observe that $\beta$ and $\beta^-$ can be of order $O(1/\varepsilon)$ to satisfy the constraints of Theorem~\ref{thm:matroid-intersection}. From now on we will suppose that $\beta, \beta^-$ are $O(1/\varepsilon)$.
	\end{remark}
	
\section{Application to One-Way Communication}
    \label{sec:communication}
    Given two matroids $\mathcal{M}_1 = (V, \mathcal{I}_1)$ and $\mathcal{M}_2 = (V, \mathcal{I}_2)$, in the one-way communication model, Alice and Bob are given $V_A$ and $V_B=V \backslash V_A$ respectively, and the goal is for Alice to send a small message to Bob so that Bob can output a large intersection of matroids $\mathcal{M}_1$ and $\mathcal{M}_2$. Here we will show that if Alice communicates an appropriate Density-Constrained Subset of $V_A$, with parameters $\beta, \beta^-$ of order $O(1/\varepsilon)$, then Bob is able to get a $3/2 + \varepsilon$ approximation of the optimal intersection.
    
    \communicationtheorem*
    
    By Theorem~\ref{thm:construction} we know that a DCS in the two restricted matroids $\mathcal{M}_1|V_A$ and $\mathcal{M}_2|V_A$ always exists, and by Proposition~\ref{prop:size-bound} we know that the number of elements sent by Alice is at most $O(\mu(V)/\varepsilon)$. Hence we only need to prove the following lemma.
    
    \begin{lemma}
        \label{lem:communication-dcs}
        Let $\varepsilon > 0$, $\beta$ and $\beta^-$ be parameters such that $\beta \geq \beta^- + 7$ and $(\beta^- - 4) \cdot (1 + \varepsilon) \geq \beta$, if $V'$ is a $(\beta, \beta^-)$-DCS of the two matroids $\mathcal{M}_1|V_A$ and $\mathcal{M}_2|V_A$, then $(3/2 + \varepsilon) \cdot \mu(V' \cup V_B) \geq \mu(V)$.
    \end{lemma}
    
    \begin{proof}
        Let $O$ be an optimal solution in $V$. Let $O_A = O \cap  V_A$ and $O_B= O \cap V_B$. Let $C_1$ and $C_2$ be sets such that $C_1 \cup C_2 = V' \cup O_B$, $C_1 \cap C_2 = \emptyset$ and they minimize the sum $\rnk_{\mathcal{M}_1}(C_1) + \rnk_{\mathcal{M}_2}(C_2)$.  By Theorem~\ref{thm:matroid-intersection} we know that $\rnk_{\mathcal{M}_1}(C_1) + \rnk_{\mathcal{M}_2}(C_2)=\mu(V' \cup O_B)$, the maximum size of a common independent set in $V' \cup O_B$.
		
		As in the proof of Theorem~\ref{thm:dcs-ratio}, we will build an auxiliary set $S$, starting with $S = \emptyset$. If there exists an element $o_1 \in O$ such that $o_1 \notin \spn_{\mathcal{M}_1}(C_1) \cup  \spn_{\mathcal{M}_2}(C_2)$, then we add $o_1$ into $S$ and we next consider the contracted matroids $\mathcal{M}_1/S$ and $\mathcal{M}_2/S$. We keep the same sets $C_1$ and $C_2$ and we try again to find an element $o_2 \in O \backslash S$ such that $o_2 \notin \spn_{\mathcal{M}_1/S}(C_1)\cup  \spn_{\mathcal{M}_2/S}(C_2)$, and we add $o_2$ to $S$. We repeat that operation until it is no longer possible to add into $S$ another element of $O$ satisfying the aforementioned constraint. Note that here all the elements $o_i$ added to $S$ come necessarily from $O_A$.
		
		As a result, as $O \backslash S$ is a common independent subset in $\mathcal{M}_1/S$ and $\mathcal{M}_2/S$, and because of Theorem~\ref{thm:matroid-intersection}, the size of $O \backslash S$ is upper-bounded by $\rnk_{\mathcal{M}_1/S}(C_1) + \rnk_{\mathcal{M}_2/S}(C_2) \leq \rnk_{\mathcal{M}_1}(C_1) + \rnk_{\mathcal{M}_2}(C_2) = \mu(V' \cup O_B)$, as in the proof of Theorem~\ref{thm:dcs-ratio}.
		
		Now we need to upper-bound the value of $|S|$. We will use the same construction as that of the proof of Theorem~\ref{thm:dcs-ratio}, as there is no difference in the algorithms that construct the sets $R_{l,i}$ and $Q_{l,i}$ (those remain subsets of $V'$, we can just follow the same procedures described in Lemmas~\ref{lem:defineR} and~\ref{lem:defineQ}). 
		
		After similar computations, we get the inequality:
		\[\beta \geq  \frac{|R|}{2\cdot |S|} + \frac{|R|}{\rnk_{\mathcal{M}_1}(R_{2}) + \rnk_{\mathcal{M}_2}(R_{1})}.\] 
		
		As the elements of $S$ are from $O_A \subset V_A$, and because of Property~(ii) of Definition~\ref{def:dcs}, we know that for all $o_i \in S$,
		\[\beta^- \leq \tilde{\rho}_{\mathcal{M}_1}(o_i) + \tilde{\rho}_{\mathcal{M}_2}(o_i) \leq |R_{1,i}| + |R_{2,i}| + 4,\]
		so by averaging over all the elements of $S$ we get
		\[\beta^- \leq \frac{|R|}{|S|} + 4.\]
		
		Therefore we derive
		\[\left(\beta - \frac{\beta^- - 4}{2}\right) \cdot (\rnk_{\mathcal{M}_1}(R_{2}) + \rnk_{\mathcal{M}_2}(R_{1})) \geq |R|.\]
		Then, as $(\beta^- - 4) \cdot |S| \leq |R|$ and $\rnk_{\mathcal{M}_1}(R_{2}) + \rnk_{\mathcal{M}_2}(R_{1}) \leq \rnk_{\mathcal{M}_1}(C_1) + \rnk_{\mathcal{M}_2}(C_2) = \mu(V' \cup O_B)$ we finally have $\left(\beta - \frac{\beta^- - 4}{2}\right) \cdot \mu(V' \cup O_B) \geq (\beta^- - 4) \cdot |S|$, and therefore:
		\[\mu(V) = |O \backslash S| + |S| \leq \left(\frac{1}{2} + \frac{\beta}{\beta^--4}\right) \cdot \mu(V' \cup O_B) \leq \left(\frac{3}{2} + \varepsilon\right) \cdot \mu(V' \cup O_B) \leq \left(\frac{3}{2} + \varepsilon\right) \cdot \mu(V' \cup V_B),\]
		as $(\beta^- - 4) \cdot (1 + \varepsilon) \geq \beta$.
    \end{proof}
	
\section{Application to Random-Order Streams}
    \label{sec:streaming}
    
    Now we consider our problem in the random-order streaming model. As our algorithm builds on that of Bernstein~\cite{bernstein:LIPIcs:2020:12419} for the unweighted simple matching, let us briefly summarize his approach. In the first phase of the streaming, he constructs a subgraph that satisfies only a weaker definition of EDCS in Definition~\ref{def:intro-edcs} (only Property (i) holds). In the second phase of the streaming, he collects the ``underfull'' edges, which are those edges that violate Property (ii). He shows that in the end, the union of the subgraph built in the first phrase and the underfull edges collected in the second phase, with high probability, contains a $3/2+\varepsilon$ approximation and that the total memory used is in the order of $O(k \cdot poly(\log(k), 1/\varepsilon))$ (there $k$ refers to the number of vertices in the graph). As we will show below, this approach can be adapted to our context of matroid intersection.
    
    \begin{definition}
        We say that a subset $V'$ has \emph{bounded density} $\beta$ if for every element $v \in V'$, $\tilde{\rho}_{\mathcal{M}_1}(v) + \tilde{\rho}_{\mathcal{M}_2}(v) \leq \beta$. 
    \end{definition}
    
    \begin{definition} 
        Let $V'$ be a subset of $V$ with bounded density $\beta$. For any parameter $\beta^-$, we say that an element $v \in V \backslash V'$ is \emph{$(V',\beta,\beta^-)$-underfull} if $\tilde{\rho}_{\mathcal{M}_1}(v) + \tilde{\rho}_{\mathcal{M}_2}(v) < \beta^-$.
    \end{definition}
    
    As in~\cite{bernstein:LIPIcs:2020:12419}, we can get a good approximation by combining a subset $V'$ of bounded density $\beta$ and the set of $(V',\beta,\beta^-)$-underfull elements in $V \backslash V'$. The proof of the following lemma is quite similar to that of Theorem~\ref{thm:dcs-ratio}, so we will only highlight the points where the proofs differ. 
    
    We begin by noting that in~\cite{bernstein:LIPIcs:2020:12419,HuangS2022-edcs}, the proof is done by showing that the combination of the subgraphs built in the first and second phase of the algorithm contains a subgraph which is an EDCS with respect to some subgraph containing the optimal solution. Our approach here is different in that we do not try to get a DCS of a well-chosen subset containing the optimal solution. Instead, we adapt directly the proof of Theorem~\ref{thm:dcs-ratio}.
    
    \begin{lemma}
        \label{lem:underfull-intersection}
        Let $\varepsilon > 0$, $\beta$ and $\beta^-$ be parameters such that $\beta \geq \beta^- + 7$ and $(\beta^- - 4) \cdot (1 + \varepsilon) \geq \beta$. Given a subset $V' \subseteq V$ with bounded density $\beta$, if $X$ contains all elements in $V \backslash V'$ that are $(V',\beta,\beta^-)$-underfull, then $(3/2 + \varepsilon) \cdot \mu(V' \cup X) \geq \mu(V)$.
    \end{lemma}
    
    \begin{proof}
        Let $O$ be an optimal solution in $V$. Let $X^{\textrm{\normalfont opt}} = X \cap O$. Let $C_1$ and $C_2$ be sets such that $C_1 \cup C_2 = V' \cup X^{\textrm{\normalfont opt}}$, $C_1 \cap C_2 = \emptyset$, and they minimize the sum $\rnk_{\mathcal{M}_1}(C_1) + \rnk_{\mathcal{M}_2}(C_2)$. By  Theorem~\ref{thm:matroid-intersection} we know that $\rnk_{\mathcal{M}_1}(C_1) + \rnk_{\mathcal{M}_2}(C_2)= \mu(V' \cup X^{\textrm{\normalfont opt}})$, the maximum size of a common independent set in $V' \cup X^{\textrm{\normalfont opt}}$.
		
		As in the proof of Theorem~\ref{thm:dcs-ratio}, we will build an auxiliary set $S$, starting with $S = \emptyset$. If there exists an element $o_1 \in O$ such that $o_1 \notin \spn_{\mathcal{M}_1}(C_1) \cup \spn_{\mathcal{M}_2}(C_2)$, then we add $o_1$ into $S$ and we now consider the contracted matroids $\mathcal{M}_1/S$ and $\mathcal{M}_2/S$. We keep the same sets $C_1$ and $C_2$ and we try again to find an element $o_2 \in O \backslash S$ such that $o_2 \notin \spn_{\mathcal{M}_1/S}(C_1) \cup \spn_{\mathcal{M}_2/S}(C_2)$, and we add $o_2$ to $S$. We repeat that operation until it is no longer possible to add another element to $S$ satisfying the aforementioned constraints.
		
		As a result, as $O \backslash S$ is a common independent subset in $\mathcal{M}_1/S$ and $\mathcal{M}_2/S$, and because of Theorem~\ref{thm:matroid-intersection}, the size of $O \backslash S$ is upper-bounded by $\rnk_{\mathcal{M}_1/S}(C_1) + \rnk_{\mathcal{M}_2/S}(C_2) \leq \rnk_{\mathcal{M}_1}(C_1) + \rnk_{\mathcal{M}_2}(C_2) = \mu(V' \cup X^{\textrm{\normalfont opt}})$, as in the proof of Theorem~\ref{thm:dcs-ratio}.
		
		Now we need to upper-bound the value of $|S|$. We will use the same construction as that of the proof of Theorem~\ref{thm:dcs-ratio}, as there is no difference in the algorithms construct the sets $R_{l,i}$ and $Q_{l,i}$ (those remain subsets of $V'$, we can just follow the same procedures described in Lemmas~\ref{lem:defineR} and~\ref{lem:defineQ}). 
		
		Then after similar computations, we get the inequality:
		\[\beta \geq  \frac{|R|}{2\cdot |S|} + \frac{|R|}{\rnk_{\mathcal{M}_1}(R_{2}) + \rnk_{\mathcal{M}_2}(R_{1})}.\] 
		
		As the elements of $S$ are not underfull (observe that here we use this fact, instead of using Property~(ii) of Definition~\ref{def:dcs} as we have done in the proof of Theorem~\ref{thm:dcs-ratio}), we know that for all $o_i \in S$,
		\[\beta^- \leq \tilde{\rho}_{\mathcal{M}_1}(o_i) + \tilde{\rho}_{\mathcal{M}_2}(o_i) \leq |R_{1,i}| + |R_{2,i}| + 4,\]
		so by averaging over all the elements of $S$ we get
		\[\beta^- \leq \frac{|R|}{|S|} + 4.\]
		
		Therefore we obtain
		\[\left(\beta - \frac{\beta^- - 4}{2}\right) \cdot (\rnk_{\mathcal{M}_1}(R_{2}) + \rnk_{\mathcal{M}_2}(R_{1})) \geq |R|.\]
		Then, as $(\beta^- - 4) \cdot |S| \leq |R|$ and $\rnk_{\mathcal{M}_1}(R_{2}) + \rnk_{\mathcal{M}_2}(R_{1}) \leq \rnk_{\mathcal{M}_1}(C_1) + \rnk_{\mathcal{M}_2}(C_2) = \mu(V' \cup X^{\mathrm{\normalfont opt}})$ we finally have $\left(\beta - \frac{\beta^- - 4}{2}\right) \cdot \mu(V' \cup X^{\mathrm{\normalfont opt}}) \geq (\beta^- - 4) \cdot |S|$, and therefore:
		\[\mu(V) = |O \backslash S| + |S| \leq \left(\frac{1}{2} + \frac{\beta}{\beta^--4}\right) \cdot \mu(V' \cup X^{\mathrm{\normalfont opt}}) \leq \left(\frac{3}{2} + \varepsilon\right) \cdot \mu(V' \cup X^{\mathrm{\normalfont opt}}),\]
		as $(\beta^- - 4) \cdot (1 + \varepsilon) \geq \beta$.
    \end{proof}
    
    Here we recall a classic probabilistic tool that we will use in the analysis of our algorithm.
    \begin{proposition}[Hoeffding's inequality] \label{prop:hoeffding}
        Let $X_1, \dots, X_t$ be $t$ negatively associated random variables that take values in $[0,1]$. Let $X := \sum_{i = 1}^tX_i$. Then, for all $\lambda > 0$ we have:
        \[\mathbb{P}(X - \mathbb{E}[X] \geq \lambda) \leq \exp\left(-\frac{2\lambda^2}{t}\right).\]
    \end{proposition}
    
    The following ideas for the streaming algorithm come from a recent paper originally intended for $b$-matchings~\cite{HuangS2022-edcs}. For sake of completeness, we reproduce the details in the following, with some slight adaptations to our more general case of matroid intersection.
    
    \begin{algorithm}[h]
	\caption{Algorithm for computing an intersection of two matroids in a random-order stream}\label{algo:intersection-streaming}
	\begin{algorithmic}[1]
	\State $V' \gets \emptyset$
	\State $\forall\,0 \leq i \leq \log_2 k,\, \alpha_i \gets \left\lfloor\frac{\varepsilon \cdot n}{\log_2(k) \cdot (2^{i+2}\beta^2 + 1)}\right\rfloor$
	\For{$i = 0 \dots \log_2 k$}
	    \State $\textsc{ProcessStopped} \gets \textsc{False}$
	    \For{$2^{i+2}\beta^2 + 1$ iterations}
	        \State $\textsc{FoundUnderfull} \gets \textsc{False}$
	        \For{$\alpha_i$ iterations}
	            \State let $v$ be the next element in the stream
	            \If{$\tilde{\rho}_{\mathcal{M}_1}(v) + \tilde{\rho}_{\mathcal{M}_2}(v) < \beta^-$}
	                \State add $v$ to $V'$
	                \State $\textsc{FoundUnderfull} \gets \textsc{True}$
                    \While{there exists $v' \in V' : \tilde{\rho}_{\mathcal{M}_1}(v') + \tilde{\rho}_{\mathcal{M}_2}(v') > \beta$}
                        \State remove $v'$ from $V'$
                    \EndWhile
	            \EndIf
	        \EndFor
	        \If{$\textsc{FoundUnderfull} = \textsc{False}$}
	            \State $\textsc{ProcessStopped} \gets \textsc{True}$
	            \State \textbf{break} from the loop
	        \EndIf
	    \EndFor
	    \If{$\textsc{ProcessStopped} = \textsc{True}$}
	        \State \textbf{break} from the loop
	    \EndIf
	\EndFor
	\State $X \gets \emptyset$
	\For{each $v$ remaining element in the stream}
	    \If{$\tilde{\rho}_{\mathcal{M}_1}(v) + \tilde{\rho}_{\mathcal{M}_1}(v) < \beta^-$}
	        \State add $v$ to $X$
	    \EndIf
	\EndFor
	\State \Return the maximum common independent set in $V' \cup X$
	\end{algorithmic}
	\end{algorithm}
    
    The algorithm, formally described in Algorithm~\ref{algo:intersection-streaming}, consists of two phases. The first phase, corresponding to Lines~3-18, constructs a subset $V'$ of bounded density $\beta$ using only an $\varepsilon$ fraction of the stream $V^{\textrm{\normalfont early}}$. In the second phase, the algorithm collects the underfull elements in the remaining part of the stream $V^{\textrm{\normalfont late}}$. As in~\cite{bernstein:LIPIcs:2020:12419} we use the idea that if no underfull element was found in an interval of size $\alpha$ (see Lines~6-13), then with high probability the number of underfull elements remaining in the stream is bounded by some value $\gamma = 4 \log(n) \frac{n}{\alpha}$. The issue is therefore how to choose the right size of interval $\alpha$, because we ignore the order of magnitude of $\mu(V)$ the optimal solution: if we do as in~\cite{bernstein:LIPIcs:2020:12419} by choosing only one fixed size of intervals $\alpha$, then if $\alpha$ is too small, the value of $\gamma$ will be too big compared to $\mu(V)$, whereas if the value of $\alpha$ is too large we will be unable to terminate the first phase of the algorithm within the early fraction of size $\varepsilon m$.
    Therefore, the idea in the first phase of the algorithm is to ``guess'' the value of $\log_2\mu(V)$ by trying successively larger and larger values of $i$ (see Line~3). In fact, by setting $i_0 = \lceil \log_2 \mu(V) \rceil$, we know that the number of insertion/deletion operations that can be performed on a $(\beta,\beta^-)$-DSC is bounded by $2^{i_0+2}\beta^2$ (see the proof of Theorem~\ref{thm:construction}). As a result we know that the first phase should always stop at a time where $i$ is smaller than or equal to $i_0$, and therefore at a time when $\alpha_i \geq \alpha_{i_0}$. Then we can prove that with high probability the number of remaining underfull elements in the stream is at most $\gamma_i = 4 \log(n) \frac{n}{\alpha_i}$.
	
	\begin{claim} \label{claim:late-part}
	    With probability at least $1 - \exp(- 2 \cdot \varepsilon^2 \cdot \mu(V))$ the late part of the stream $V^{\textrm{\normalfont late}}$ contains at least a $(1 - 2 \varepsilon)$ fraction of the optimal solution. Moreover, in expectation $V^{\textrm{\normalfont late}}$ contains a $(1 - \varepsilon)$ fraction of the optimal solution.
	\end{claim}
	
	\begin{proof}
	    Consider an optimal solution $O = \{o_1,\ldots, o_{\mu(V)}\}$. We define the random variables $X_i = \mathbbm{1}_{o_i \in V^{\textrm{\normalfont early}}}$. Hence we have $\mathbb{E}[\sum X_i] = \varepsilon \cdot |O|$. Moreover, the random variables $X_i$ are negatively associated, so we can use Hoeffding's inequality (see Proposition~\ref{prop:hoeffding}) to get
	    \[\mathbb{P}\left[\sum_{i = 1}^{\mu(V)} X_i \geq 2 \varepsilon \cdot \mu(V)\right] \leq \exp\left(-\frac{2 \cdot \varepsilon^2 \cdot \mu(V)^2}{\mu(V)}\right) = \exp\left(-2 \cdot \varepsilon^2 \cdot \mu(V)\right).\]
	\end{proof}
	
	Recall that we defined $i_0 = \lceil \log_2 \mu(V) \rceil$. Algorithm~\ref{algo:intersection-streaming} works when $\mu(V)$ is not too big (otherwise we may use intervals of size $\alpha_{i_0} = \lfloor \frac{\varepsilon\cdot n}{\log_2(k) \cdot (2^{i_0+2}\beta^2 + 1)} \rfloor = 0$). Here we will first argue that this case can be handled anyway. 
	
	\begin{claim} \label{claim:big-output}
	    We can assume that $\frac{\varepsilon\cdot n}{\log_2(k) \cdot (2^{i_0+2}\beta^2 + 1)} \geq 1$.
	\end{claim}
	
	\begin{proof}
	    If this is not the case, then we can just store all the elements of $V$ as the number of elements $n$ is bounded by $\frac{\log_2(k) \cdot (2^{i_0+2}\beta^2 + 1)}{\varepsilon} = O(\mu(V) \cdot \log(k) \cdot (1/\varepsilon)^3)$ (as $\beta$ is $O(1/\varepsilon)$, see Remark~\ref{rmk:val-eps}). As a result, if at some point of the first phase we have not stopped and we have $\alpha_i = 0$, then we store all the remaining elements of $V^{\textrm{\normalfont late}}$ and we will be able to get a $(1 - \varepsilon)$ approximation in expectation and a $(1 - 2\varepsilon)$ approximation with high probability (more precisely, at least $1 - \exp(-2 \cdot c \cdot \varepsilon^5 \cdot n/\log(k))$, for some constant $c > 0$, see Claim~\ref{claim:late-part}), using $O(\mu(V) \cdot \log(k) \cdot (1/\varepsilon)^3)$ memory.
	\end{proof}
	
	From now on we will assume that $\frac{\varepsilon\cdot n}{\log_2(k) \cdot (2^{i_0+2}\beta^2 + 1)} \geq 1$.
	Then we can move on to our main algorithm. The following lemma is very similar to the one used in~\cite{bernstein:LIPIcs:2020:12419}.
	
	\begin{lemma} \label{lem:early-properties}
	    The first phase of Algorithm~\ref{algo:intersection-streaming} uses $O(\beta\cdot \mu(V))$ memory and constructs a subset $V' \subseteq V$, satisfying the following properties:
	    \begin{enumerate}
	        \item The first phase terminates within the first $\varepsilon \cdot n$ elements of the stream.
	        \item When the first phase terminates after processing some element, we have:
	            \begin{enumerate}
	                \item $V'$ has bounded density $\beta$, and  contains at most $O(\beta \cdot \mu(V))$ elements.
	                \item With probability at least $1 - n^{-3}$, the total number of $(V', \beta, \beta^-)$-underfull elements in the remaining part of the stream is at most $\gamma = O( \mu(V) \cdot \log(n) \cdot \log(k) \cdot \beta^2 \cdot 1/\varepsilon)$.
	            \end{enumerate}
	    \end{enumerate}
	\end{lemma}

	\begin{proof}
        First, in each interval of size $\alpha_i$ processed until the first phase terminates (except the last interval), at least one insertion/deletion operation that is performed (as described in the proof of Theorem~\ref{thm:construction}), and therefore the total number of such processed intervals is bounded by $2 \beta^2\cdot \mu(V) + 1$. As a result, the first phase ends with some $i \leq i_0 = \lceil \log_2 \mu(V) \rceil$, and the total number of elements processed in the first phase is therefore bounded by $\varepsilon \cdot n \cdot \frac{i_0}{\log_2(k)} \leq \varepsilon \cdot n$. For Property~2.a, as the subset $V'$ built always keeps a bounded density $\beta$, Proposition~\ref{prop:size-bound} implies that $V'$ uses $O(\beta \cdot \mu(V)) = O(\mu(V) \cdot 1/\varepsilon)$ memory.

        Now we turn to the last property. As mentioned previously, the intuition is simple: the algorithm only exits the first phase if it fails to find a single underfull element in an entire interval (Line~14-16), and since the stream is random, such an event implies that there are most likely few underfull elements left in the stream.

        To formalize this, we call the $j$-th time that Lines 7-13 are processed the \emph{epoch} $j$. Let $\mathcal{A}_j$ be the event that $\textsc{FoundUnderfull}$ is set to $\textsc{False}$ in epoch $j$. Let $\mathcal{B}_j$ be the event that the number of $(V',\beta, \beta^-)$-underfull elements in the remaining part of the stream is larger than some $\gamma$. Note that the last property fails to hold if and only if we have $\mathcal{A}_j \land \mathcal{B}_j$ for some $j$, so we want to upper bound $\mathbb{P}[\mathcal{A}_j \land \mathcal{B}_j]$. Let $V_j^r$ contains all elements in $V$ that have not yet appeared in the stream at the \emph{beginning} of epoch $j$ (r for remaining). Let $V^e_j$ be the elements that appear in epoch $j$ (e for epoch), and note that $E^e_j$ is a subset of size $\alpha_i \geq \alpha_{i_0} = \alpha_{\lceil\log_2\mu(V)\rceil} = \alpha$ chosen uniformly at random from $V^r_j$. Define $V'_j$ to be the subset $V'$ at the beginning of epoch $j$, and define $V^u_j \subseteq E^r_j$ to be the set of remaining underfull elements with respect to $V'_j$, $\beta$, and $\beta^-$. Observe that because of event $\mathcal{A}_j$, the subset $V'$ remains the same throughout epoch $j$, so an element that is underfull at any point during the epoch will be underfull at the end as well. Thus, $\mathcal{A}_j \land \mathcal{B}_j$ is equivalent to the event that $|V^u_j| > \gamma$ and $V^u_j \cap V^e_j = \emptyset$. 

        Let $\mathcal{A}^k_j$ be the event that the $k$-th element of epoch $j$ is not in $V^u_j$. We have that $\mathbb{P}[\mathcal{B}_j \land \mathcal{A}_j] \leq \mathbb{P}[\mathcal{A}_j \,|\, \mathcal{B}_j] \leq \mathbb{P}[\mathcal{A}^1_j \, | \, \mathcal{B}_j] \prod_{k=2}^\alpha \mathbb{P}[\mathcal{A}^k_j \, | \, \mathcal{B}_j, \mathcal{A}^1_j, \ldots, \mathcal{A}^{k-1}_j],$ where the second inequality comes from that $V^e_j$ is of size larger or equal to $\alpha = \alpha_{\lceil\log_2\mu(V)\rceil}$.

        Now, observe that $\mathbb{P}[\mathcal{A}^1_j \,|\, \mathcal{B}_j] < 1 - \frac{\gamma}{n}$ because the first element of the epoch is chosen uniformly at random from the set of $\leq n$ remaining elements, and the event fails if the chosen element is in $V^u_j$, where $|V^u_j| > \gamma$ by definition of $\mathcal{B}_j$. Similarly, for any $k$, $\mathbb{P}[\mathcal{A}^k_j \, | \, \mathcal{B}_j, \mathcal{A}^1_j, \ldots, \mathcal{A}^{k-1}_j] < 1 - \frac{\gamma}{m}$ because conditioning on the previous events $\mathcal{A}^t_j$ implies that no element from $V^u_j$ has yet appeared in this epoch, so there remain still at least $\gamma$ element from $V^u_j$ left in the stream.

        We now set
        \[\gamma = 4\log(n)\cdot \frac{n}{\alpha} = 4\log(n)\cdot n \cdot \left\lfloor\frac{\varepsilon\cdot n}{\log_2(k) \cdot (2^{i_0+2}\beta^2 + 1)}\right\rfloor^{-1},\]
        and as we assumed that $\frac{\varepsilon\cdot n}{\log_2(n) \cdot (2^{i_0+2}\beta^2 + 1)} \geq 1$ (and as a factor of at most $2$ separates $\lfloor x \rfloor$ and $x$ when $x \geq 1$) we have $\gamma = O(\mu(V)\cdot \log(n) \cdot \log(k) \cdot (1/\varepsilon)^3)$.
        
        Combining the above equations yields that $\mathbb{P}[\mathcal{B}_j \land \mathcal{A}_j] \leq (1-\frac{\gamma}{n})^{\alpha} = (1 - \frac{4\log(n)}{\alpha})^\alpha \leq n^{-4}$. There are clearly at most $n$ epochs, so union bounding over all of them shows that the last property fails with probability at most $n^{-3}$, as desired.
    \end{proof}
    
    Then we can combine the previous results to obtain the following theorem.
	
	\streamingtheorem*
	
	\begin{proof}
	    Using Lemma~\ref{lem:underfull-intersection} on the graph $V' \cup V^{\textrm{\normalfont late}}$ we get $(3/2 + \varepsilon) \cdot \mu(V' \cup X) \geq \mu(V' \cup V^{\textrm{\normalfont late}})$.
	    Applying Claim~\ref{claim:late-part}, we know that in expectation $(1 - \varepsilon)^{-1} \cdot \mu(V' \cup V^{\textrm{\normalfont late}}) \geq \mu(V)$. Hence in expectation we also have \[(3/2 + \varepsilon) \cdot (1 - \varepsilon)^{-1} \cdot \mu(V' \cup X) \geq \mu(V).\]
	    Moreover, by Lemma~\ref{lem:early-properties}, the memory consumption is bounded by $O(\mu(V)\cdot \log(n) \cdot \log(k) \cdot (1/\varepsilon)^3)$ with probability at least $1 - n^{-3}$. Hence we can decide that, if during the execution of the algorithm at some point the memory consumption reaches the bound defined in Lemma~\ref{lem:early-properties} (recall that this bound can be computed as it depends only on the epoch when the first phase stopped), then we discard the remaining elements. As this event happens only with probability $1 - n^{-3}$, this is not harmful for the expectation of the approximation ratio.
	    
	    Moreover, using Claim~\ref{claim:late-part}, we know that a $(1 - 2\varepsilon)^{-1}\cdot (3/2 + \varepsilon)$ approximation of the optimal common independent subset is contained in $\mu(V \cup X)$ with probability at least $1 - \exp(-2 \cdot \varepsilon^2 \cdot \mu(V))$. As the memory consumption of $O(\mu(V)\cdot \log(n) \cdot \log(k) \cdot (1/\varepsilon)^3)$ is guaranteed with probability at least $1 - n^{-3}$ (see Lemma~\ref{lem:early-properties}), then with probability at least $1 - (\exp(- 2 \cdot \varepsilon^2 \cdot \mu(V)) + n^{-3})$ (by union bound), we can obtain a $(1 - 2\varepsilon)^{-1}\cdot (3/2 + \varepsilon)$ approximation using $O(\mu(V)\cdot \log(n) \cdot \log(k) \cdot (1/\varepsilon)^3)$ memory. As $\varepsilon < 1/4$, we have $(1 - 2\varepsilon)^{-1}\cdot (3/2 + \varepsilon) \leq (3/2 + 8\varepsilon)$, and therefore to get a $3/2 + \varepsilon$, we have to use an $\varepsilon' = \varepsilon/8$ so that the probability to have an approximation ratio worse than $3/2 + \varepsilon$ approximation is at most $\exp(-2 \cdot (\varepsilon/8)^2 \cdot \mu(V)) + n^{-3}$.
	\end{proof}


	
\appendix

\section{Deferred Proofs}
    \label{app:proof}

    \begin{proof}[Proof of Lemma~\ref{lem:add-increase}]
        In the following, we will use the notation $P_{a,b} = U_a^{\textrm{\normalfont old}} \cap U_b^{\textrm{\normalfont new}}$ for $a$, $b \in \llbracket 1, k \rrbracket$. 
        
	    We prove (i) by strong induction. We start with the case $j = 1$. Let $i_0$ be the largest $i$ set such as $P_{1,i} \neq \emptyset$. Then we know that for all $v \in U_1^{\textrm{\normalfont old}}$,
	    \begin{align*}
	        \tilde{\rho}_{\mathcal{M}}^{\textrm{\normalfont new}}(v) 
	        &\geq \rho_{\mathcal{M}'^{\textrm{\normalfont new}}/(\bigcup_{i = 1}^{i_0-1}U_i^{\textrm{\normalfont new}})}(U_{i_0}^{\textrm{\normalfont new}}) &\text{by Proposition~\ref{prop:non-increasing}}\\
	        &\geq \rho_{\mathcal{M}'^{\textrm{\normalfont new}}/(\bigcup_{i = 1}^{i_0-1}U_i^{\textrm{\normalfont new}})}(P_{1,i_0}) &\text{by maximality of $U_{i_0}^{\textrm{\normalfont new}}$}\\
	        &\geq \rho_{\mathcal{M}'^{\textrm{\normalfont new}}/(\bigcup_{i = 1}^{i_0-1}P_{1,i})}(P_{1,i_0}) & \text{as $\bigcup_{i = 1}^{i_0-1}P_{1,i} \subseteq \bigcup_{i = 1}^{i_0-1}U_i^{\textrm{\normalfont new}}$ (Proposition~\ref{prop:subset_increased_rho})}\\
	        &= \rho_{\mathcal{M}'^{\textrm{\normalfont old}}/(\bigcup_{i = 1}^{i_0-1}P_{1,i})}(P_{1,i_0}) & \text{as $u^{\textrm{\normalfont new}} \notin U_1^{\textrm{\normalfont old}}$}\\
	        &\geq \rho_{\mathcal{M}'^{\textrm{\normalfont old}}}(U_1^{\textrm{\normalfont old}}).  & \text{by Proposition~\ref{prop:largerdensity}}
	    \end{align*}
	    
	    For the induction step, let $2 \leq j \leq k$. Suppose that the property is true for all $i \in \llbracket 1, j - 1 \rrbracket$. We want to prove that the property is also true for $j$. Let $i_0$ be the largest $i$ such that $P_{j, i} \neq \emptyset$. Then we know that for all $v \in U_j^{\textrm{\normalfont old}}$,
	    \begin{align*}
	        \tilde{\rho}_{\mathcal{M}}^{\textrm{\normalfont new}}(v) &\geq \rho_{\mathcal{M}'^{\textrm{\normalfont new}}/(\bigcup_{i = 1}^{i_0-1}U_i^{\textrm{\normalfont new}})}(U_{i_0}^{\textrm{\normalfont new}}) \geq \rho_{\mathcal{M}'^{\textrm{\normalfont new}}/(\bigcup_{i = 1}^{i_0-1}U_i^{\textrm{\normalfont new}})}(P_{j,i_0}).
	    \end{align*}
	    Then we have two cases:
	    \begin{itemize}
	        \item We have $\bigcup_{i = 1}^{j-1}U_i^{\textrm{\normalfont old}} \subseteq \bigcup_{i = 1}^{i_0-1}U_i^{\textrm{\normalfont new}}$. In that case we can write, similarly as in the previous case,
	        \begin{align*}
    	        \tilde{\rho}_{\mathcal{M}}^{\textrm{\normalfont new}}(v) &\geq \rho_{\mathcal{M}'^{\textrm{\normalfont new}}/(\bigcup_{i = 1}^{i_0-1}U_i^{\textrm{\normalfont new}})}(P_{j,i_0})\\
    	        &\geq \rho_{\mathcal{M}'^{\textrm{\normalfont new}}/(\bigcup_{i = 1}^{j-1}U_i^{\textrm{\normalfont old}} \cup \bigcup_{i = 1}^{i_0-1}P_{j,i})}(P_{j,i_0}) &\text{as $\bigcup_{i = 1}^{j-1}U_i^{\textrm{\normalfont old}} \cup \bigcup_{i = 1}^{i_0-1}P_{j,i} \subseteq \bigcup_{i = 1}^{i_0-1}U_i^{\textrm{\normalfont new}}$}\\
    	        &= \rho_{\mathcal{M}'^{\textrm{\normalfont old}}/(\bigcup_{i = 1}^{j-1}U_i^{\textrm{\normalfont old}} \cup \bigcup_{i = 1}^{i_0-1}P_{j,i})}(P_{j,i_0}) & \text{as $u^{\textrm{\normalfont new}} \notin \bigcup_{i = 1}^{j-1}U_i^{\textrm{\normalfont old}} \cup \bigcup_{i = 1}^{i_0}P_{j,i}$}\\
    	        &\geq \rho_{\mathcal{M}'^{\textrm{ old}}}(U_j^{\textrm{\normalfont old}}). &  \text{by Proposition~\ref{prop:largerdensity}}
    	    \end{align*}
	        \item Otherwise, there exists $u' \in U_{i_1}^{\textrm{\normalfont old}}$ such that $i_1 < j$ and $u' \notin \bigcup_{i = 1}^{i_0-1}U_i^{\textrm{\normalfont new}}$. It means that $u' \in U_{i_2}^{\textrm{\normalfont new}}$ for some $i_2 \geq i_0$, and hence $\tilde{\rho}_{\mathcal{M}}^{\textrm{\normalfont new}}(u') \leq \tilde{\rho}_{\mathcal{M}}^{\textrm{\normalfont new}}(v)$.
	        Then we have 
	        
	        $$\tilde{\rho}_{\mathcal{M}}^{\textrm{\normalfont new}}(v) \geq  \tilde{\rho}_{\mathcal{M}}^{\textrm{\normalfont new}}(u') \geq \rho_{\mathcal{M}'^{\textrm{\normalfont old}}/(\bigcup_{i = 1}^{i_1-1}U_i^{\textrm{\normalfont old}})}(U_{i_1}^{\textrm{\normalfont old}}) \geq 
	        \rho_{\mathcal{M}'^{\textrm{\normalfont old}}/(\bigcup_{i = 1}^{j-1}U_i^{\textrm{\normalfont old}})}(U_{j}^{\textrm{\normalfont old}}),
	        $$
	        where the second inequality uses the strong induction hypothesis and the third uses Proposition~\ref{prop:non-increasing}. 
	    \end{itemize}
	    This concludes the proof of (i).
	    
	    Now we move to (ii). Observe that (i) implies the result for $v \in V'^{\textrm{\normalfont old}}$. If $v \notin \spn_{\mathcal{M}}(V'^{\textrm{\normalfont old}})$, then it is clear that the density associated to that element can only increase (recall that by Definition~\ref{def:tilderho}, $\tilde{\rho}_{\mathcal{M}}^{\textrm{\normalfont old}}(v)=0$). From now on we suppose that $v \in \spn_{\mathcal{M}}(V'^{\textrm{\normalfont old}})$ and we denote $j = \min\{j \in \llbracket 1, k \rrbracket : v \in \spn_{\mathcal{M}}(\bigcup_{i = 1}^{j}U_i^{\textrm{\normalfont old}})\}$. By (i) we know that for all $v' \in \bigcup_{i = 1}^{j}U_i^{\textrm{\normalfont old}}$, we have $\tilde{\rho}_{\mathcal{M}}^{\textrm{\normalfont new}}(v') \geq \rho_{\mathcal{M}'^{\textrm{\normalfont old}}/(\bigcup_{i = 1}^{j-1}U_i^{\textrm{\normalfont old}})}(U_{j}^{\textrm{\normalfont old}}) = \tilde{\rho}_{\mathcal{M}}^{\textrm{\normalfont old}}(v') \geq \tilde{\rho}_{\mathcal{M}}^{\textrm{\normalfont old}}(v)$, hence we also have $\tilde{\rho}_{\mathcal{M}}^{\textrm{\normalfont new}}(v) \geq \min_{v' \in \bigcup_{i = 1}^{j}U_i^{\textrm{\normalfont old}}}\tilde{\rho}_{\mathcal{M}}^{\textrm{\normalfont old}}(v') \geq \tilde{\rho}_{\mathcal{M}}^{\textrm{\normalfont old}}(v)$, as the associated density of $v$ will be at least equal to the smallest density of the elements in $\bigcup_{i = 1}^{j}U_i^{\textrm{\normalfont old}}$, as once all those elements are in the decomposition we are sure that $v$ is spanned.
	    
	    For (iii), we only have to prove the upper bound (the lower bound comes from (ii)). Suppose that in the new decomposition, $u^{\textrm{\normalfont new}}$ appears in $U_{i_0}^{\textrm{\normalfont new}}$. Then it means that for all $i < i_0$, we have $U_i^{\textrm{\normalfont new}} = U_i^{\textrm{\normalfont old}}$, as the previous sets have been made using the very same elements.
	    
	    If $u^{\textrm{\normalfont new}} \notin \spn_{\mathcal{M}/(\bigcup_{i = 1}^{i_0-1}U_i^{\textrm{\normalfont old}})}(U_{i_0}^{\textrm{\normalfont new}}\backslash \{u^{\textrm{\normalfont new}}\})$, then the only possibility is $\rho_{\mathcal{M}'^{\textrm{\normalfont new}}/(\bigcup_{i = 1}^{i_0-1}U_i^{\textrm{\normalfont old}})}(U_{i_0}^{\textrm{\normalfont new}}) = 1$ and we have our upper bound. From now on we assume that $u^{\textrm{\normalfont new}} \in \spn_{\mathcal{M}/(\bigcup_{i = 1}^{i_0-1}U_i^{\textrm{\normalfont old}})}(U_{i_0}^{\textrm{\normalfont new}}\backslash \{u^{\textrm{\normalfont new}}\})$. 
	    
	    Let $j = \min\{j \in \llbracket 1, k \rrbracket : u^{\textrm{\normalfont new}} \in \spn_{\mathcal{M}}(\bigcup_{i = 1}^{j}U_i^{\textrm{\normalfont old}})\}$ 
	    ($j$ is well-defined by our assumption 
	    immediately before). 
	    Let $P_{i,i_0} = U_i^{\textrm{\normalfont old}} \cap U_{i_0}^{\textrm{\normalfont new}}$. Let $i_1$ be the largest $i$ such that $P_{i,i_0} \neq \emptyset$. As $u^{\textrm{\normalfont new}} \in \spn_{\mathcal{M}}(\bigcup_{i = 1}^{i_0-1}U_i^{\textrm{\normalfont old}} \cup \bigcup_{i = 1}^{i_1}P_{i, i_0}) \subseteq \spn_{\mathcal{M}}(\bigcup_{i = 1}^{i_1}U_i^{\textrm{\normalfont old}})$ 
	    (by observing that $i_1 \geq i_0$ as 
	    $U^{\textrm{\normalfont new}}_i = U^{\textrm{\normalfont old}}_i$ for $i < i_0$),  
	    it means that $i_1 \geq j$. As a result we have
	    \begin{align*}
	        \tilde{\rho}_{\mathcal{M}}^{\textrm{\normalfont new}}(u^{\textrm{\normalfont new}})
	        &=\rho_{\mathcal{M}'^{\textrm{\normalfont new}}/(\bigcup_{i = 1}^{i_0-1}U_i^{\textrm{\normalfont old}})}(U_{i_0}^{\textrm{\normalfont new}})\\
	        &\leq \rho_{\mathcal{M}'^{\textrm{\normalfont new}}/(\bigcup_{i = 1}^{i_0-1}U_i^{\textrm{\normalfont old}} \cup \bigcup_{i = 1}^{i_1-1}P_{i,i_0})}(P_{i_1, i_0} \cup \{u^{\textrm{\normalfont new}}\}) & 
	        \text{by Proposition~\ref{prop:largerdensity}}
	        \\
	        &\leq \rho_{\mathcal{M}'^{\textrm{\normalfont new}}/(\bigcup_{i = 1}^{i_0-1}U_i^{\textrm{\normalfont old}} \cup \bigcup_{i = 1}^{i_1-1}P_{i,i_0})}(P_{i_1, i_0}) + 1\\
	        &= \rho_{\mathcal{M}'^{\textrm{\normalfont old}}/(\bigcup_{i = 1}^{i_0-1}U_i^{\textrm{\normalfont old}} \cup \bigcup_{i = 1}^{i_1-1}P_{i,i_0})}(P_{i_1, i_0}) + 1 &\text{as $u^{\textrm{\normalfont new}} \notin \bigcup_{i = 1}^{i_0-1}U_i^{\textrm{\normalfont old}} \cup \bigcup_{i = 1}^{i_1}P_{i,i_0}$}\\
	        &\leq \rho_{\mathcal{M}'^{\textrm{\normalfont old}}/(\bigcup_{i = 1}^{i_1-1}U_i^{\textrm{\normalfont old}})}(P_{i_1, i_0}) + 1 & 
	        \text{by Proposition~\ref{prop:subset_increased_rho}}\\
	        &\leq \rho_{\mathcal{M}'^{\textrm{\normalfont old}}/(\bigcup_{i = 1}^{i_1-1}U_i^{\textrm{\normalfont old}})}(U_{i_1}^{\textrm{\normalfont old}}) + 1 &\text{by maximality of $U_{i_1}^{\textrm{\normalfont old}}$}\\
	        &\leq \rho_{\mathcal{M}'^{\textrm{\normalfont old}}/(\bigcup_{i = 1}^{j-1}U_i^{\textrm{\normalfont old}})}(U_{j}^{\textrm{\normalfont old}}) + 1 &\text{by Proposition~\ref{prop:non-increasing} and $i_1\geq j$}\\
	        &= \tilde{\rho}_{\mathcal{M}}^{\textrm{\normalfont old}}(u^{\textrm{\normalfont new}}) + 1.
	    \end{align*}
	    
	    For property (iv), we can first observe that the elements having densities larger than $\tilde{\rho}_{\mathcal{M}}^{\textrm{\normalfont new}}(u^{\textrm{\normalfont new}})$, which is upper bounded 
	    by $ \tilde{\rho}_{\mathcal{M}}^{\textrm{\normalfont old}}(u^{\textrm{\normalfont new}}) + 1$, will remain with the same densities (actually they will even remain in the same sets $U_i$ as it was observed in the proof of (iii)). So we just focus on the case $\tilde{\rho}_{\mathcal{M}}^{\textrm{\normalfont old}}(v) < \tilde{\rho}_{\mathcal{M}}^{\textrm{\normalfont old}}(u^{\textrm{\normalfont new}})$. As $\tilde{\rho}_{\mathcal{M}}^{\textrm{\normalfont old}}(u^{\textrm{\normalfont new}}) > 0$, we can set $j = \min\{j \in \llbracket 1, k \rrbracket : u^{\textrm{\normalfont new}} \in \spn_{\mathcal{M}}(\bigcup_{i = 1}^{j}U_i^{\textrm{\normalfont old}})\}$. 
	    
	    Let $U_{\textrm{\normalfont small}} = \bigcup_{i = j+1}^{k}U_i^{\textrm{\normalfont old}}$ and $U_{\textrm{\normalfont big}} = \bigcup_{i = 1}^{j}U_i^{\textrm{\normalfont old}}$. First, suppose that there exists an element $v \in U_{\textrm{\normalfont small}}$ such that $\tilde{\rho}_{\mathcal{M}}^{\textrm{\normalfont new}}(v) \geq \tilde{\rho}_{\mathcal{M}}^{\textrm{\normalfont old}}(u^{\textrm{\normalfont new}})$. Let $P = \{v \in  U_{\textrm{\normalfont small}} : \tilde{\rho}_{\mathcal{M}}^{\textrm{\normalfont new}}(v) \geq \tilde{\rho}_{\mathcal{M}}^{\textrm{\normalfont old}}(u^{\textrm{\normalfont new}})\}$. Let $P_i = P \cap U_i^{\textrm{\normalfont new}}$ for all $i$, and let $i_0$ be the smallest index $i$ such that $P_i \neq \emptyset$. Then 
	    \begin{align*}
	        \tilde{\rho}_{\mathcal{M}}^{\textrm{\normalfont new}}(v) 
	        &\leq \rho_{\mathcal{M}'^{\textrm{\normalfont new}}/(\bigcup_{i = 1}^{i_0-1}U_i^{\textrm{\normalfont new}})}(U_{i_0}^{\textrm{\normalfont new}}) &\text{by Proposition~\ref{prop:non-increasing}}\\
	        &\leq \rho_{\mathcal{M}'^{\textrm{\normalfont new}}/(\bigcup_{i = 1}^{i_0}U_i^{\textrm{\normalfont new}} \backslash P_{i_0})}(P_{i_0}) &\text{by Proposition~\ref{prop:largerdensity}}\\
	        &\leq \rho_{\mathcal{M}'^{\textrm{\normalfont new}}/(U_{\textrm{\normalfont big}} \cup \{u^{\textrm{\normalfont new}}\})}(P_{i_0}) &\text{as $\bigcup_{i = 1}^{i_0}U_i^{\textrm{\normalfont new}} \backslash P_{i_0} \subseteq \bigcup_{i = 1}^{j}U_i^{\textrm{\normalfont old}} \cup \{u^{\textrm{\normalfont new}}\}$}\\
	        &= \rho_{\mathcal{M}'^{\textrm{\normalfont old}}/(U_{\textrm{\normalfont big}})}(P_{i_0}) &\text{as $u^{\textrm{\normalfont new}} \in \spn_{\mathcal{M}}(U_{\textrm{\normalfont big}})$}\\
	        &\leq \rho_{\mathcal{M}'^{\textrm{\normalfont old}}/(U_{\textrm{\normalfont big}})}(U_{j+1}^{\textrm{\normalfont old}}) &\text{by maximality of $U_{j+1}^{\textrm{\normalfont old}}$}\\
	        &< \rho_{\mathcal{M}'^{\textrm{\normalfont old}}/(U_{\textrm{\normalfont big}} \backslash U_{j}^{\textrm{\normalfont old}})}(U_{j}^{\textrm{\normalfont old}}) & \text{by Proposition~\ref{prop:non-increasing}}\\
	        &= \tilde{\rho}_{\mathcal{M}}^{\textrm{\normalfont old}}(u^{\textrm{\normalfont new}}),
	    \end{align*}
	    a contradiction to the assumption that $i_0$ exists, implying that $P =\emptyset$. 
	    
	    As a result, the set of elements in $V'$ of density smaller than $\tilde{\rho}_{\mathcal{M}}^{\textrm{\normalfont old}}(u^{\textrm{\normalfont new}})$ remains the same (note that this set cannot become bigger because of property (i)). Thereby, the execution of Algorithm~\ref{algo:decomposition} can be decomposed into two phases, the early phase when the elements of $U_{\textrm{\normalfont big}} \cup \{u^{\textrm{\normalfont new}}\}$ are processed, and the late phase when the elements of $U_{\textrm{\normalfont small}}$ are processed. As $\spn_{\mathcal{M}}(U_{\textrm{\normalfont big}} \cup \{u^{\textrm{\normalfont new}}\}) = \spn_{\mathcal{M}}(U_{\textrm{\normalfont big}})$, the construction of the sets in the late phase is the same no matter $u^{\textrm{\normalfont new}}$ is in $V'$ or not. Hence the sets are the same and so are the associated densities. This concludes the proof of (iv).
	\end{proof}
	
	\begin{proof}[Proof of Lemma~\ref{lem:del-decrease}]
        Consider the behavior when $u^{\textrm{\normalfont old}}$ is added to $V'\backslash \{u^{\textrm{\normalfont old}}\}$: it is clear that Lemma~\ref{lem:add-increase} applies. As a result the points (i) and (ii) come easily from Lemma~\ref{lem:add-increase} (ii). For (iii), observe that from Lemma~\ref{lem:add-increase} (iii) we get that $\tilde{\rho}_{\mathcal{M}}^{\textrm{\normalfont new}}(u^{\textrm{\normalfont old}}) \leq \tilde{\rho}_{\mathcal{M}}^{\textrm{\normalfont old}}(u^{\textrm{\normalfont old}}) \leq \tilde{\rho}_{\mathcal{M}}^{\textrm{\normalfont new}}(u^{\textrm{\normalfont old}}) + 1$ and hence we obtain also (iii) here. For (iv) the bounds are a bit different from what we could get from Lemma~\ref{lem:add-increase} (iv) but using ideas similar to that from the previous proof one can easily show the desired result.
    \end{proof}

    \begin{proof}[Proof of Theorem~\ref{thm:construction}]
        Start with an empty subset $V'$. Then apply the following local improvement steps repeatedly on $V'$, until it is no longer possible. If an element in $V'$ violates Property~(i) of Definition~\ref{def:dcs}, then remove it from $V'$; similarly, if an element in $V \backslash V'$ violates Property~(ii), insert it into $V'$. Note that among the two local improvement steps, the priority is given to the deletion operations.
        
        Observe that when no element violates Property~(i), all the elements have densities bounded by $\beta$ in both matroids. To prove that this algorithm terminates in finite time and to show the existence of a DCS, we introduce a potential function:
        \[\Phi(V') = (2\beta - 7) \cdot |V'| - \sum_{l \in\{1,2\}} \left[\sum_{j = 1}^k \left(\rnk_{\mathcal{M}'_l/(\bigcup_{i = 1}^{j-1}U_{l,i})}(U_{l,j}) \cdot (\rho_{\mathcal{M}'_l/(\bigcup_{i = 1}^{j-1}U_{l,i})}(U_{l,j}))^2\right)\right]\]
        where $U_{l,1}, \dots, U_{l,k}$ denotes the density-based decomposition of $V'$ in $\mathcal{M}_l$ for $l \in \{1,2\}$. We can rewrite this function in a more convenient form:
        \[\Phi(V') = (2\beta - 7) \cdot |V'| - \sum_{l \in\{1,2\}} \left[\sum_{j = 1}^k \rho_{l, j}^2\right]\]
        where for $l \in \{1, 2\}$, the vector $\rho_l = (\rho_{l,1},\dots, \rho_{l,k})$ is the list of the densities of each set of the decomposition $U_{l,1}, \dots, U_{l,k}$ counted with multiplicity equal to their rank (so that, for instance, $\rho_{\mathcal{M}'_l/(\bigcup_{i = 1}^{j-1}U_i)}(U_j)$ appears $\rnk_{\mathcal{M}'_l/(\bigcup_{i = 1}^{j-1}U_i)}(U_j)$ times in that vector; we potentially add some zeros in the end so that the vector has exactly $k$ components).
        
        The execution of the algorithm can be seen as a series of batches of operations, consisting of one insertion operation followed by some number of deletion operations. Each batch has a finite size because we can make only a finite number of deletions when no insertion is performed. At the end of each batch of operations, all the densities are bounded by $\beta$, hence using  Proposition~\ref{prop:size-bound} (as for this result to hold it is only required that the densities are bounded by $\beta$) we have that $\Phi$ is bounded by $(2\beta - 7) \cdot \beta \cdot \mu(V)$. Then we have to show that $\Phi$ increases at each local improvement step by at least some constant amount and we will be done.
        
        When Property~(i) of Definition~\ref{def:dcs} is not satisfied by some element in $u^{\textrm{\normalfont old}} \in V'^{\textrm{\normalfont old}}$, then it is removed to get a new set $V'^{\textrm{\normalfont new}} = V'^{\textrm{\normalfont old}} \backslash \{u^{\textrm{\normalfont old}}\}$. Hence from the vectors $\rho_l^{\textrm{\normalfont old}} = (\rho_{l,1},\dots, \rho_{l,k})$ we get new vectors $\rho_l^{\textrm{\normalfont new}} = (\rho_{l,1} - \lambda_{l,1},\dots, \rho_{l,k} - \lambda_{l,k})$, with the following properties:
        \begin{itemize}
            \item $\lambda_{l,i} \geq 0$ (by Lemma~\ref{lem:del-decrease} (ii));
            \item $\sum_{j = 1}^k\lambda_{l,j} = 1$ for $l \in \{1,2\}$ (as we always have $\sum_{j = 1}^k \rho_{l,j} = |V'|$, see Proposition~\ref{prop:sum-density});
            \item $\lambda_{l,i} > 0 \Rightarrow \tilde{\rho}_{\mathcal{M}_l}^{\textrm{\normalfont old}}(u^{\textrm{\normalfont old}}) - 1 \leq \rho_{l,i} \leq \tilde{\rho}_{\mathcal{M}_l}^{\textrm{\normalfont old}}(u^{\textrm{\normalfont old}})$ for $l \in \{1,2\}$ (by Lemma~\ref{lem:del-decrease} (iv)).
        \end{itemize}
        As a result we get:
        \begin{align*}
            \Phi(V'^{\textrm{\normalfont new}}) - \Phi(V'^{\textrm{\normalfont old}}) &= - (2 \beta - 7) + \sum_{l \in\{1,2\}} \left[\sum_{j = 1}^k \rho_{l, j}^2 - (\rho_{l, j} - \lambda_{l,j})^2\right]\\
            &= - 2 \beta + 7 + \sum_{l \in\{1,2\}} \left[\sum_{j = 1}^k 2\rho_{l, j}\lambda_{l,j} - \sum_{j=1}^k \lambda_{l,j}^2\right]\\
            &\geq - 2 \beta + 5 + \sum_{l \in\{1,2\}} \left[\sum_{j = 1}^k 2\rho_{l, j}\lambda_{l,j}\right]\\
            &= - 2 \beta + 5 + \sum_{l \in\{1,2\}} \left[\sum_{\tilde{\rho}_{\mathcal{M}_l}^{\textrm{\normalfont old}}(u^{\textrm{\normalfont old}}) - 1 \leq \rho_{l,j} \leq \tilde{\rho}_{\mathcal{M}_l}^{\textrm{\normalfont old}}(u^{\textrm{\normalfont old}})} 2\rho_{l, j}\lambda_{l,j}\right]\\
            &\geq -2\beta + 5 + \sum_{l\in \{1,2\}}\left[2 \cdot (\tilde{\rho}_{\mathcal{M}_l}^{\textrm{\normalfont old}}(u^{\textrm{\normalfont old}}) - 1)\sum_{\tilde{\rho}_{\mathcal{M}_l}^{\textrm{\normalfont old}}(u^{\textrm{\normalfont old}}) - 1 \leq \rho_{l,j} \leq \tilde{\rho}_{\mathcal{M}_l}^{\textrm{\normalfont old}}(u^{\textrm{\normalfont old}})} \lambda_{l,j}\right]\\
            &= -2\beta + 5 + 2 \cdot (\tilde{\rho}_{\mathcal{M}_1}^{\textrm{\normalfont old}}(u^{\textrm{\normalfont old}}) +  \tilde{\rho}_{\mathcal{M}_2}^{\textrm{\normalfont old}}(u^{\textrm{\normalfont old}})- 2)\\
            &> -2\beta + 1 + 2 \beta = 1.
        \end{align*}
        The first inequality comes from $\lambda_{l,i} \geq 0$ and $\sum_{j=1}^k\lambda_{l,j} = 1$, implying that $\sum_{j=1}^k\lambda_{l,j}^2 \leq 1$.
        The last inequality comes from $\tilde{\rho}_{\mathcal{M}_1}^{\textrm{\normalfont old}}(u^{\textrm{\normalfont old}}) + \tilde{\rho}_{\mathcal{M}_2}^{\textrm{\normalfont old}}(u^{\textrm{\normalfont old}}) > \beta$. Hence we get an increase of $\Phi$ of at least $1$.
        
        Similarly, when Property~(ii) of Definition~\ref{def:dcs} is not satisfied by some element in $u^{\textrm{\normalfont new}} \in V \backslash V'^{\textrm{\normalfont old}}$, then it is added to get a new set $V'^{\textrm{\normalfont new}} = V'^{\textrm{\normalfont old}} \cup \{u^{\textrm{\normalfont new}}\}$. Hence from the vectors $\rho_l^{\textrm{\normalfont old}} = (\rho_{l,1},\dots, \rho_{l,k})$ we get new vectors $\rho_l^{\textrm{\normalfont new}} = (\rho_{l,1} + \lambda_{l,1},\dots, \rho_{l,k} + \lambda_{l,k})$, with the following properties:
        \begin{itemize}
            \item $\lambda_{l,i} \geq 0$ (by Lemma~\ref{lem:add-increase} (ii));
            \item $\sum_{j = 1}^k\lambda_{l,j} = 1$ for $l \in \{1,2\}$ (as we always have $\sum_{j = 1}^k \rho_{l,j} = |V'|$, see Proposition~\ref{prop:sum-density});
            \item $\lambda_{l,i} > 0 \Rightarrow \tilde{\rho}_{\mathcal{M}_l}^{\textrm{\normalfont old}}(u^{\textrm{\normalfont new}}) \leq \rho_{l,i} \leq \tilde{\rho}_{\mathcal{M}_l}^{\textrm{\normalfont old}}(u^{\textrm{\normalfont new}}) + 1$ for $l \in \{1,2\}$ (by Lemma~\ref{lem:add-increase} (iv)).
        \end{itemize}
        As a result we get:
        \begin{align*}
            \Phi(V'^{\textrm{\normalfont new}}) - \Phi(V'^{\textrm{\normalfont old}}) &= (2 \beta - 7) - \sum_{l \in\{1,2\}} \left[\sum_{j = 1}^k (\rho_{l, j}^2 + \lambda_{l,j})^2 - \rho_{l, j}^2\right]\\
            &= 2 \beta - 7 - \sum_{l \in\{1,2\}} \left[\sum_{j = 1}^k 2\rho_{l, j}\lambda_{l,j} + \sum_{j=1}^k \lambda_{l,j}^2\right]\\
            &\geq 2 \beta - 9 - \sum_{l \in\{1,2\}} \left[\sum_{j = 1}^k 2\rho_{l, j}\lambda_{l,j}\right]\\
            &= 2 \beta - 9 - \sum_{l \in\{1,2\}} \left[\sum_{\tilde{\rho}_{\mathcal{M}_l}^{\textrm{\normalfont old}}(u^{\textrm{\normalfont new}}) \leq \rho_{l,j} \leq \tilde{\rho}_{\mathcal{M}_l}^{\textrm{\normalfont old}}(u^{\textrm{\normalfont new}}) + 1} 2\rho_{l, j}\lambda_{l,j}\right]\\
            &\geq 2\beta - 9 - \sum_{l\in \{1,2\}}\left[2 \cdot (\tilde{\rho}_{\mathcal{M}_l}^{\textrm{\normalfont old}}(u^{\textrm{\normalfont new}}) + 1)\sum_{\tilde{\rho}_{\mathcal{M}_l}^{\textrm{\normalfont old}}(u^{\textrm{\normalfont new}}) \leq \rho_{l,j} \leq \tilde{\rho}_{\mathcal{M}_l}^{\textrm{\normalfont old}}(u^{\textrm{\normalfont new}}) + 1} \lambda_{l,j}\right]\\
            &= 2\beta - 9 - 2 \cdot (\tilde{\rho}_{\mathcal{M}_1}^{\textrm{\normalfont old}}(u^{\textrm{\normalfont new}}) + \tilde{\rho}_{\mathcal{M}_2}^{\textrm{\normalfont old}}(u^{\textrm{\normalfont new}}) + 2)\\
            &> 2\cdot (\beta - \beta^-) - 13 \geq 1.
        \end{align*}
        The first inequality comes from $\lambda_{l,i} \geq 0$ and $\sum_{j=1}^k\lambda_{l,j} = 1$, implying that $\sum_{j=1}^k\lambda_{l,j}^2 \leq 1$.
        To move to the last line we use that $\tilde{\rho}_{\mathcal{M}_1}^{\textrm{\normalfont old}}(u^{\textrm{\normalfont new}}) + \tilde{\rho}_{\mathcal{M}_2}^{\textrm{\normalfont old}}(u^{\textrm{\normalfont new}}) < \beta^-$, and then that $\beta \geq \beta^- + 7$. Hence we also get an increase of $\Phi$ of at least $1$.
        
        As a result, a $(\beta, \beta^-)$-DCS can be found in at most $2\cdot \beta^2 \cdot \mu(V)$ such local improvement steps.
    \end{proof}

    \bibliography{library}
    
\end{document}